\documentclass[showpacs,amsmath,amssymb,nofootinbib,prd]{revtex4}
\usepackage{ifpdf}
\ifpdf
  \usepackage[%
    colorlinks=true,%
    hypertexnames=false,%
    backref=page,%
    citecolor=blue,%
    pdftitle={Asymptotics of Fixed Point Distributions for Inexact Monte Carlo
      Algorithms},%
    pdfdisplaydoctitle,%
    pdfauthor={M. A. Clark and A. D. Kennedy},%
    pdfkeywords={Algorithms, Quantum Field Theory, Lattice QCD},%
    pdfstartview={FitH},%
    pdfpagemode={FullScreen}]{hyperref}
\fi
\font\tenbb=msbm10
\font\sevenbb=msbm7
\font\fivebb=msbm5
\newfam\bbfam
\def\bb{\fam\bbfam\tenbb}
\textfont\bbfam=\tenbb
\scriptfont\bbfam=\sevenbb
\scriptscriptfont\bbfam=\fivebb
\def\sqr#1#2{{\vcenter{\hrule height.#2pt
   \hbox{\vrule width.#2pt height#1pt \kern#1pt
      \vrule width.#2pt}
   \hrule height.#2pt}}}

\def\bsqr#1#2{{\vrule width #1pt height#2pt}}
\def\bsquare{{\mathchoice\bsqr66\bsqr66\bsqr33\bsqr33}}
\def\Sest{{\Sigma}}		    
\def\L{{\cal L}}		    
\def\A{{\cal A}}		    
\def\T{{\cal T}}		    
\def\I{{\cal I}}		    
\def\M{{\cal M}}		    
\def\dt{{\delta\tau}}		    
\def\trjlen{\tau}		    
\def\DS{{\Delta S}}		    
\def\dH{{\delta H}}		    
\def\dJ{{\tr\ln U_*}}		    
\def\dbar{{\bar\delta}}
\def\twovector[#1,#2]{\left(\begin{array}{c} #1 \\ #2 \end{array}\right)}
\def\comp{\circ}		    
\def\badbreak{\penalty1000}
\newtheorem{theorem}{Theorem}
\newtheorem{lemma}{Lemma}
\newtheorem{definition}{Definition}
\newtheorem{corollary}{Corollary}
\newenvironment{proof}{{\par\em Proof.}}{\badbreak$\;\bsquare$\smallskip}
\def\asymp{\sim}		    
\def\tr{\mathop{\rm tr}}	    
\def\ad{\mathop{\rm ad}}	    
\def\deg{\mathop{\rm deg}}	    
\def\dd#1#2{{\mathchoice{\partial#1\over\partial#2}%
  {\partial#1\over\partial#2}%
  {\partial#1\!/\!\partial#2}%
  {\partial#1\!/\!\partial#2}}}	    
\def\ddq{\dd{}q}		    
\def\ddp{\dd{}p}		    
\def\ddd#1#2#3{{\mathchoice{\partial^2#1\over\partial#2\partial#3}%
  {\partial^2#1\over\partial#2\partial#3}%
  {\partial^2#1\!/\!\partial#2\partial#3}%
  {\partial^2#1\!/\!\partial#2\partial#3}}} 
\def\implies{\Rightarrow}	    
\def\defn{\equiv}		    
\def\rational#1#2{{\mathchoice{\textstyle{#1\over#2}}%
  {\scriptstyle{#1\over#2}}{\scriptscriptstyle{#1\over#2}}{#1/#2}}}
\def\half{\rational12}		    
\def\third{\rational13}		    
\def\quarter{\rational14}	    
\def\rmsub#1#2{#1_{\mbox{\tiny #2}}} 
\def\Q{{\bb Q}}			    
\def\R{{\bb R}}			    
\def\Z{{\bb Z}}			    
\def\M{\mathcal{M}}		    
\def\Seff{\mathop{S_{\rm eff}}}	    
\def\Nf{\rmsub{N}{f}}		    
\def\Sf{\rmsub{S}{F}}		    
\def\Sg{\rmsub{S}{G}}		    

\def\deltaH{\delta H}

\def\minv{\M^{-1}}
\def\dmdu{\frac{\partial\M}{\partial U}}

\begin{document}

\title{Asymptotics of Fixed Point Distributions for Inexact Monte Carlo
  Algorithms}

\author{M. A. Clark}
  \email{mikec@bu.edu}
  \affiliation{%
    Center for Computational Sciences, Boston University,\\
    3 Cummington Street, Boston, MA 02215, United States of America}%
\author{A. D. Kennedy}
  \email{adk@ph.ed.ac.uk}
  \affiliation{
    School of Physics, The University of Edinburgh,\\
    Mayfield Road, Edinburgh, EH9 3JZ, United Kingdom}%
\date{\today}
\begin{abstract}
  \noindent We introduce a simple general method for finding the equilibrium
  distribution for a class of widely used inexact Markov Chain Monte Carlo
  algorithms.  The explicit error due to the non-commutivity of the updating
  operators when numerically integrating Hamilton's equations can be derived
  using the Baker--Campbell--Hausdorff formula.  This error is manifest in the
  conservation of a ``shadow'' Hamiltonian that lies close to the desired
  Hamiltonian.  The fixed point distribution of inexact Hybrid algorithms may
  then be derived taking into account that the fixed point of the momentum
  heatbath and that of the molecular dynamics do not coincide exactly.  We
  perform this derivation for various inexact algorithms used for lattice QCD
  calculations.  {\parfillskip=0pt\par}
\end{abstract}

\maketitle

\section{Introduction}

The first algorithms used to simulate lattice QCD with dynamical quarks fell
into three classes: Monte Carlo \cite{pseudofermions,Weingarten:1981hx},
Langevin \cite{batrouni85a}, and Molecular Dynamics (microcanonical)
\cite{callaway82a,polonyi83a} methods.  The first two involve taking many
essentially independent small random steps, leading to a dynamical critical
exponent \(z=2\) typical of Brownian motion, where this exponent relates the
autocorrelation time of the Markov process to the correlation length of the
physical system.  The steps are small either in the sense of making only a
small global change in the gauge field or by making a larger change locally,
such as updating a single link variable.  The third suffers from the problem of
not being ergodic in general, and therefore not necessarily generating the
correct distribution of configurations; nevertheless, although the updates are
also built out of many small steps these steps are correlated so as to lead to
a dynamical critical exponent \(z=1\), and it can thus be considered as a
``large step'' algorithm.  The last two algorithms also depend upon a step size
parameter \(\dt\), and are only ``exact'' in the limit \(\dt\to0\).

It was then realized \cite{duane85a} that by combining Molecular Dynamics with
a momentum refreshment heatbath an ergodic large step ``Hybrid'' algorithm
resulted, although it still suffered from step size errors in the equilibrium
distribution.  Shortly afterwards it was found that by combining such Hybrid
updates with a Metropolis Monte Carlo acceptance test these step size errors
could be completely eliminated \cite{duane87a}: this is called the Hybrid Monte
Carlo algorithm.  Both the idea of combining Molecular Dynamics with momentum
refreshment \cite{andersen80a} and that of correcting the Langevin algorithm by
a Metropolis test \cite{rossky78a} had been introduced previously in other
fields.

The introduction of the Hybrid Monte Carlo algorithm did not end the use of
inexact algorithms for two reasons: firstly some people believe that the volume
dependence of the cost of the Hybrid Monte Carlo algorithm outweighs the
advantage of its exactness, and secondly they wanted to simulate with what are
now called ``rooted staggered quarks'' for which there is no explicit local
form for the action, but instead the square- or fourth-root of the
corresponding determinant is required.  A way of doing this with step size
errors of \(O(\dt^2)\) was introduced with the \(R\) algorithm
\cite{gottlieb87a}.  Only recently has the RHMC algorithm
\cite{Clark:2006wp,Clark:2006fx} provided an efficient exact alternative.

The purpose of the present paper is to analyze the inexact algorithms mentioned
above, as well as several other interesting variants.  The approach we take is
new, namely we view all the algorithms as a combination of a Molecular Dynamics
trajectory with a momentum refreshment; from this point of view the Langevin
algorithm is just a single step Molecular Dynamics trajectory.  Although in
principle, any numerical integration scheme could be used for the Molecular
Dynamics integration in practice all the algorithms use symmetric symplectic
integrators (or closely related integrators).  Such integrators have several
remarkable properties, such as being reversible and area preserving, which are
required for the exact Hybrid Monte Carlo algorithm.  We make use of the
remarkable property that all symplectic integrators exactly conserve a ``shadow
Hamiltonian'' close to the desired one for small enough step sizes to produce a
uniform and simple asymptotic expansion for the equilibrium (fixed point)
distribution of the corresponding inexact Markov chains.

We then turn to the class of inexact algorithms that use a noisy estimate of
the force due to the fermionic determinant.  We introduce the pseudofermionic
\(\chi\) algorithm as an aid to establishing that the \(R\) algorithm, which
may be considered as an interpolation between the \(\chi\) and \(R_0\)
algorithms, has only \(O(\dt^2)\) errors in its equilibrium distribution.  We
also establish that one our original motivations --- namely trying to find an
variant of the \(R\) algorithm with errors falling as a higher power of the
step size --- is almost certainly doomed to failure as the unwanted noise
contributions to the error are necessarily positive.

This paper is organised as follows: in \(\S\)\ref{sec:sympl} we use the
Baker--Campbell--Hausdorff (BCH) formula to derive the form of the ``shadow''
Hamiltonian that is conserved by various symplectic integrators.  A brief
description of Hybrid stochastic algorithms is given in \(\S\)\ref{sec:hybrid},
and the resulting fixed points of these algorithms are derived in
\(\S\)\ref{sec:equilibrium}.  We then turn our attention to noisy algorithms,
specifically as applied to fermion theories, in \(\S\)\ref{sec:noise}.  For the
reader's convenience we collect as appendices derivations of some important
results that are rather difficult to find in an accessible form in the
literature.

\section{Symplectic Integrators}
\label{sec:sympl}

We are interested in finding the classical trajectory in phase space of a
system described by a Hamiltonian \(H(q,p) = T(p) + S(q) = \half p^2 + S(q)\).
It has been known for a long time that the leapfrog integration scheme has many
desirable properties, and higher-order generalizations have been discovered
several times in different fields~\cite{campostrini89a,creutz89a,%
channel90a,yoshida90a,forest90a,suzuki90a,suzuki91a,gladman91a,sexton92a,%
yoshida93a,hut95a}.  The basic idea of such a {\em symplectic integrator} is to
write the time evolution operator as
\begin{eqnarray*}
  \exp\left(\trjlen{d\over dt}\right)
  &\defn& \exp\left(\trjlen\left\{\dd pt\ddp + \dd qt\ddq\right\}\right)\\
  &=& \exp\left(\trjlen\left\{- \dd Hq\ddp + \dd Hp\ddq\right\}\right)
  \defn e^{\trjlen\hat H} \\
  &=& \exp\left(\trjlen\left\{-S'(q)\ddp + T'(p)\ddq\right\}\right),
\end{eqnarray*}
where \(\hat H\) is the {\em Hamiltonian vector field}.  In the second step we
have made use of Hamilton's equations, and thus implicitly the fundamental
2-form\footnote{For simplicity of presentation we shall consider here only the
case where the fundamental symplectic 2-form is \(\omega=dq\wedge dp\).  For
gauge theories we need to use \(\omega=-\sum_i d(p_i\theta_i)\) where the
\(\theta_i\) are the left-invariant Maurer--Cartan forms on a Lie group
manifold~\cite{kennedy88b}, but this generalization is straightforward.}
\(\omega=dq\wedge dp\).

From the structure of our Hamiltonian, namely the fact that the kinetic energy
\(T\) is a function only of \(p\) and the potential energy \(S\) is a function
only of~\(q\), it follows that \(\exp\left[-\trjlen S'(q)\ddp\right]\) and
\(\exp\left[\trjlen T'(p)\ddq\right]\) are trivial\footnote{For gauge fields
the \(p\) lie in a Lie algebra and the \(q\) in the corresponding Lie group, so
the evaluation is straightforward if not entirely trivial.} to evaluate.  Let
us write
\begin{equation}
  Q \defn T'(p) \ddq
    \qquad\hbox{and}\qquad P \defn -S'(q) \ddp,
  \label{eq:PQdef}
\end{equation}
so that by Taylor's theorem
\begin{eqnarray}
  e^{tQ}: f(q,p) & \mapsto & f\left(q + tT'(p), p\right),\qquad \nonumber \\
  e^{tP}: f(q,p) & \mapsto & f\left(q, p - tS'(q)\right);
  \label{eq:PQmap}
\end{eqnarray}
then from the BCH formula (\ref{eq:bch}) we find that the {\em QP} symplectic
integrator leads to the time evolution given by
\begin{eqnarray*}
  \lefteqn{ \left(e^{\dt Q}e^{\dt P}\right)^{\trjlen/\dt} } & & \\
  & = &\left(\exp\left[
        (Q+P)\dt + \half[Q,P]\dt^2 + O(\dt^3)
      \right]\right)^{\trjlen/\dt} \qquad \\
    &=& \exp\left[
          \trjlen\left((Q+P) + \half[Q,P]\dt + O(\dt^2)\right)
        \right]\\
    &\defn& e^{\trjlen\widehat{H'}} = e^{\trjlen(Q+P)} + O(\dt).
\end{eqnarray*}
In addition to conserving energy to \(O(\dt)\) such symplectic integrators are
manifestly area preserving.

The BCH formula tells us more than that there is an area-preserving approximate
integration scheme for Hamilton's equations: it tells us that for each
symplectic integrator there exists a ``shadow'' Hamiltonian \(H'\) close to the
original one, which is {\em exactly conserved}.  For the {\em QP} integrator
the actual trajectories through phase space are integral curves of the vector
field\footnote{Observe that \(H'\) is linear in \(\frac{\partial}{\partial q}\)
and \(\frac{\partial}{\partial p}\) because all but that first term in the BCH
expansion are commutators.}
\begin{widetext}
\begin{eqnarray*}
  \widehat{H'} &=& Q + P + \half[Q,P]\,\dt
      + \rational1{12}\Bigl\{[Q,[Q,P]] - [P,[Q,P]]\Bigr\}\,\dt^2
      - \rational1{24}[P,[Q,[Q,P]]] \dt^3 \\
    && 
      \begin{array}{lrlrlrll}
        + \rational1{720}\Bigl\{
        & - 4 & [P,[Q,[Q,[Q,P]]]]
        & - 6 & [[Q,P],[Q,[Q,P]]]
        & + 4 & [P,[P,[Q,[Q,P]]]] & \\
        & - 2 & [[Q,P],[P,[Q,P]]]
        & - & [Q,[Q,[Q,[Q,P]]]]
	& + & [P,[P,[P,[Q,P]]]]
        & \Bigr\} \,\dt^4 + O(\dt^5)
      \end{array} \\
    &=& \left\{p\ddq - S'\ddp\right\}
      + \half\left\{S'\ddq - pS''\ddp\right\}\dt
      + \rational1{12}\left\{2pS''\ddq - (p^2S'''+2S'S'')
      \ddp\right\}\dt^2 \\
    &&\quad + \rational1{12}\left\{S'S''\ddq
        - p(S'S'''+{S''}^2)\ddp\right\}\dt^3 
      \quad + \rational1{720}\Biggl\{
        4\left(-p^3S^{(4)} + p(6{S''}^2+3S'S''')\right)\ddq \\
    &&\qquad\qquad
        +\left(p^4S^{(5)} - p^2(30S'''S''+6S^{(4)}S')
          - 12S'(2{S''}^2+S'S''')\right)\ddp\Biggr\}\dt^4 + O(\dt^5)
\end{eqnarray*}
\end{widetext}
from equation~(\ref{eq:bch}).  In Appendix~\ref{sec:poisson} we show that the
vector field \(\widehat{H'}\) is a Hamiltonian vector field, that is there
exists a Hamiltonian \(H'\) such that \[\widehat{H'} = \dd{H'}p\ddq -
\dd{H'}q\ddp,\] and that \(H'\) may be obtained explicitly by replacing the
commutators in the BCH expansion (\ref{eq:bch}) of \(\ln(e^Te^S)\) by Poisson
brackets\footnote{The Poisson brackets define a Lie algebra since they are
manifestly antisymmetric and satisfy the Jacobi relation \(\{A,\{B,C\}\} +
\{B,\{C,A\}\} + \{C,\{A,B\}\} = 0\), q.v., Appendix~\ref{sec:poisson}.} of the
form \[\{A,B\} \defn \dd{A}p\dd{B}q - \dd{A}q\dd{B}p;\] so
\begin{widetext}
\begin{eqnarray*}
  H' &=& H + \half\{T,S\}\,\dt
    + \rational1{12}\Bigl(\{T,\{T,S\}\} - \{S,\{T,S\}\}\Bigr)\,\dt^2 -\rational1{24}\{S,\{T,\{T,S\}\}\} \dt^3 \\
    && \begin{array}{lrlrll}
         + \rational1{720} \Bigl(
         & - 4 & \{S,\{T,\{T,\{T,S\}\}\}\}
         & - 6 & \{\{T,S\},\{T,\{T,S\}\}\} & \\
         \vphantom{\Bigl(}
         & + 4 & \{S,\{S,\{T,\{T,S\}\}\}\}
         & - 2 & \{\{T,S\},\{S,\{T,S\}\}\} & \\
         & - & \{T,\{T,\{T,\{T,S\}\}\}\}
         & + & \{S,\{S,\{S,\{T,S\}\}\}\}
         & \Bigr) \,\dt^4 + O(\dt^5)
      \end{array} \\
    &=& H + \half pS'\,\dt + \rational1{12}\left(p^2S''+{S'}^2\right)\,\dt^2
      + \rational1{12}pS'S''\,\dt^3  \\
      & & + \rational1{720}\left(-p^4S^{(4)} + p^2(6S'S'''+12{S''}^2)
      + 12{S'}^2S''\right)\,\dt^4 + O(\dt^5).
\end{eqnarray*}
Note that \(H'\) cannot be written as the sum of a \(p\)-dependent kinetic term
and a \(q\)-dependent potential term.  This means that any hope that one could
exactly ``invert'' this relation to find a Hamiltonian whose approximate
integral exactly conserves \(H\) is forlorn.

\subsection{Symmetric Integrators}

It is immediately apparent that we can do better than this by using a symmetric
symplectic integrator (\ref{eq:bchsym}) which just gives the familiar {\em PQP}
leapfrog scheme for which
\begin{displaymath}
  \begin{array}{rl}
    U_0(\dt)^{\trjlen/dt}
    &= \left(e^{\dt\,P/2}e^{\dt\,Q}e^{\dt\,P/2}\right)^{\trjlen/\dt} \\[0.8ex]
    &= \left(\exp\left[
          (P+Q)\,\dt - \rational1{24}\Bigl([P,[P,Q]] + 2[Q,[P,Q]]\Bigr)\,\dt^3
            + O(\dt^5)\right]\right)^{\trjlen/\dt} \\[0.8ex]
    &= \exp\left[
            \trjlen\left((P+Q)
              - \rational1{24}\Bigl([P,[P,Q]] + 2[Q,[P,Q]]\Bigr)\,\dt^2
                + O(\dt^4)\right)\right] \\[0.8ex]
    &= e^{\trjlen(P+Q)} + O(\dt^2).
  \end{array}
\end{displaymath}
The {\em PQP} integrator exactly conserves the Hamiltonian
\[  H' = H + \rational1{24}\left\{2p^2S''-{S'}^2\right\}\,\dt^2
+ \rational1{720}\left\{-p^4S^{(4)} + 6p^2(S'S'''+2{S''}^2)
- 3{S'}^2S''\right\}\,\dt^4 + O(\dt^6),
\label{eq:evol}\]
whereas the {\em QPQ} integrator conserves
\[ H' = H + \rational1{24}\left\{-p^2S''+2{S'}^2\right\}\,\dt^2 
+ \rational1{5760}\left\{7p^4S^{(4)} - 24p^2(3S'S'''+{S''}^2)
+ 96{S'}^2S''\right\}\,\dt^4 + O(\dt^6).\]
\end{widetext}
It is possible to construct symplectic integrators of arbitrarily high order;
some aspects of this theory are discussed in Appendix~\ref{sec:hossi}.

\section{Hybrid Stochastic Algorithms}
\label{sec:hybrid}

\subsection{The Hybrid algorithm}

The Hybrid algorithm~\cite{duane85a,duane85b,duane86b}, so called because it is
a hybrid of the Langevin \cite{batrouni85a} and Molecular Dynamics algorithms,
constructs a Markov process on a ``fictitious'' phase space consisting of the
field variables of interest (the coordinates) and a set of corresponding
``fictitious'' momenta.  It should be emphasized that these momenta have
nothing to do with the momenta which occur in the field equations of motion and
canonical quantization relations; they are just introduced to define a
classical dynamics of the fields in a new ``fictitious'' time dimension.  If we
are considering a four dimensional field theory, for example, then our
``fictitious'' dynamics takes place in a new fifth time dimension.  To this end
we introduce a Hamiltonian \(H(q,p)\) and a fundamental symplectic \(2\)-form
as in section~\ref{sec:sympl}.  We then iterate three Markov steps, each of
which has the distribution \(e^{-H}\) as an approximate fixed point and which,
when taken together, are ergodic.\footnote{In some cases we should choose
trajectory lengths from some random distribution for this to be true.}
\begin{itemize}

\item The first such step is {\em momentum refreshment} whereby the fictitious
momenta are chosen from a Gaussian heatbath.\footnote{A generalization of this
is the ``partial momentum refreshment'' \cite{kennedy91a} step used in the
second-order Langevin \cite{horowitz87,horowitz90a} or Kramer's algorithm
\cite{kuti88b,beccaria94a,beccaria94b}.}

\item The second step is to integrate Hamilton's equations using an approximate
integrator for some length of ficititious time~\(\trjlen\), and then to flip
the momenta
\begin{displaymath}
  \twovector[q,p] \mapsto U(\trjlen)\twovector[q,p] \defn \twovector[q'',p'']
    \mapsto \twovector[q'',-p''].
\end{displaymath}
It has the distribution \(e^{-H'}\) as an exact fixed point, and it thus
approximately preserves the desired Hamiltonian \(H\).  If the integrator is
symplectic then this step is area-preserving, and if it is also symmetric then
the step is reversible (this is why we incorporate the final momentum flip).

\item The third step is a {\em momentum flip}, \(F:p\mapsto-p\), which
obviously preserves the Gaussian distribution of momenta.  Since this step is
immediately followed by a momentum refreshment it can be ignored in all cases
but Kramer's algorithm \cite{horowitz87,horowitz90a,kuti88b,beccaria94a,%
beccaria94b}.
\end{itemize}

\subsection{Langevin Algorithm}

The Langevin equation \cite{langevin08a,batrouni85a} is just a special case of
the Hybrid algorithm for a single leapfrog integration step with {\em PQP}
ordering: in other words we use the symplectic integrator \(e^{\dt P/2} e^{\dt
Q} e^{\dt P/2}\) using the notation of equation~(\ref{eq:PQdef}).  This
operator maps \((q,p)\mapsto(q',p')\) where \(q' = q + p\,\dt - \half
S'(q)\dt^2\), which is the familiar Langevin equation when the the Gaussian
distributed momenta \(p\) are written as Gaussian noise \(\eta\) and the time
step as \(\epsilon\defn\dt^2\):
\begin{displaymath}
  q' = q - \half S'(q)\epsilon + \eta\sqrt\epsilon.
\end{displaymath}
The {\em PQP} symplectic integrator has the amusing property of giving rise to
a simple closed form Langevin equation, but there are many other variants
too.  For example, the {\em QPQ} symplectic integrator \(e^{\dt Q/2} e^{\dt P}
e^{\dt Q/2}\) leads to the equation
\begin{widetext}
\begin{equation}
  q' = q + p\,\dt - \half S'\dt^2 - \quarter S''p\dt^3 -
  \rational1{16}S'''p^2\dt^4 - \rational1{96}S^{(4)}p^3\dt^5 -
  \rational1{768}S^{(5)}p^4\dt^6 + O(\dt^7),
  \label{eq:QPQevol}
\end{equation}
the lowest-order {\em PQP} Campostrini wiggle (q.v., Appendix~\ref{sec:hossi})
gives
\begin{eqnarray*}
  q' &=& q + p\,\dt - \half S'\dt^2 - \rational16S''p\dt^3
      + \rational1{24}[S''S'-S'''p^2]\dt^4 \\
    && \begin{array}{rl}
         \rational1{432}\Bigl[\left(
           6(3 + \root3\of2 + \root3\of4)S'''S'
           - 3(4 + 4\root3\of2 + 3\root3\of4){S''}^2\right)p \\
           - (3 + \root3\of2 + \root3\of4)S^{(4)}p^3\Bigr]\dt^5 + O(\dt^6)
      \end{array}
\end{eqnarray*}
and the second-order {\em PQP} Campostrini wiggle gives
\begin{eqnarray*}
  q'& = &q + p\,\dt - \half S'\dt^2 - \rational16S''p\dt^3
      + \rational1{24}[S''S'-S'''p^2]\dt^4 \\
    &&+\rational1{120}\left[
      \left(3S'''S' + {S''}^2\right)p - S^{(4)}p^3\right]\dt^5 \\
    &&+ \rational1{720}\left[
      - (3S'''S' + {S''}^2)S'
      + \left(5S'''S'' + 6S^{(4)}S'\right)p^2
      - S^{(5)}p^4\right]\dt^6 + O(\dt^7)
\end{eqnarray*}
which agrees with the exact evolution operator \(e^{\dt\,\hat H}\) up to
\(O(\dt^7)\).
\end{widetext}

\section{Equilibrium Distributions}
\label{sec:equilibrium}

We now want to address the question as to what fixed point (equilibrium)
distribution is produced by the Langevin \cite{langevin08a,fokker14a,planck17a}
and Hybrid algorithms.  The existence of such a fixed point distribution and
its uniqueness follow from their ergodicity.  We are, of course, looking for a
fixed point distribution of \(q\) alone, and not of both \(q\) and \(p\).

Since in general the momentum dependence of \(H'\) is not exactly Gaussian we
do not expect either \(e^{-H}\) or \(e^{-H'}\) to be an exact fixed point of a
Hybrid Markov process.  As we expect the equilibrium distribution to be close
to the desired one, \(e^{-S}\), it seems reasonable to parameterize it as
\(e^{-(S+\DS)}\).  The condition that this is a fixed point is
\begin{eqnarray*}
\lefteqn{ e^{-\bigl(S(q') + \DS(q')\bigr)} } \qquad & &\\
    & = & \int dq\, e^{-\bigl(S(q)+\DS(q)\bigr)} \int dp\, e^{-\half p^2}
      \delta(q' - q'')\\
    & = & \int dq\,dp\, e^{ - H(q,p) - \DS(q)} \delta(q' - q''),
\end{eqnarray*}
where \(\twovector[q'',p''] = U(\trjlen)\twovector[q,p]\) is the phase space
point reached at the end of a trajectory of length~\(\trjlen\).

To solve this equation we change variables to \((q'',p'')\), whence we obtain
\begin{eqnarray*}
\lefteqn{e^{-\bigl(S(q') + \DS(q')\bigr)} } \qquad & & \\
& = & \int dq''\,dp''\, (\det U_*)^{-1} e^{- (H
+ \DS)\comp U^{-1}} \delta(q' - q'').
\end{eqnarray*}
If \(U\) is area-preserving then its Jacobian \(\det U_* \defn\det
\dd{(q'',p'')}{(q,p)} \defn e^\dJ = 1\), and if it is reversible then
\(U^{-1} = F\comp U\comp F\) where \(F\) is the momentum-flip
operation \(F: p\mapsto-p\).

In the general case we introduce the operation \(\dbar: \Omega
\mapsto\Omega\comp U^{-1} - \Omega\comp F\comp U\comp F\), which is zero if
\(U\) is reversible, so \[(H + \DS)\comp U^{-1} = (H + \DS)\comp F\comp U\comp
F + \dbar(H + \DS).\] Next we note that \(H\) is an even function of \(p\),
namely \(H\comp F=H\), and \(\DS\) does not depend on \(p\), so \(\DS\comp
F=\DS\) also, hence \[(H + \DS)\comp U^{-1} = (H + \DS)\comp U\comp F + \dbar(H
+ \DS).\] In terms of the operator \(\delta:\Omega\mapsto \Omega\comp U\comp F
- \Omega\comp F\), which measures the change in some quantity \(\Omega\) over a
trajectory,\footnote{An extra momentum flip is included for convenience.}
\begin{eqnarray*}
  \lefteqn{ (H + \DS)\comp U^{-1} } \qquad & & \\
  &=& (H + \DS)\comp F + \delta(H+\DS) + \dbar(H+\DS) \\
  &=& H + \DS + (\delta + \dbar)(H+\DS).
\end{eqnarray*}
Therefore
\begin{eqnarray*}
  \lefteqn{ e^{-\bigl(S(q') + \DS(q')\bigr)} } & & \\
    &=\int dq''\,dp''\, e^{- (H + \DS)} e^{- (\delta + \dbar)(H + \DS)
      - \dJ} \delta(q' - q'') & \\
    &=e^{-\bigl(S(q') + \DS(q')\bigr)} 
      \int dp''\, e^{-\half p''^2} e^{- (\delta + \dbar)(H+\DS) - \dJ}, &
\end{eqnarray*}
and thus we obtain the condition
\begin{equation}
  \left\langle e^{- (\delta + \dbar)(H + \DS) - \dJ} \right\rangle_p = 1.
\label{eq:req}
\end{equation}

Since \(H\) is extensive so is \(\dH\), and thus we can show order by order in
\(\dt\) that \(\DS\) is extensive too.

\subsection{Langevin Algorithm}

For a reversible and area-preserving integration scheme, such as a symmetric
symplectic integrator, the equilibrium distribution must satisfy
\begin{equation}
  \left\langle e^{-\delta(H + \DS)} \right\rangle_p = 1.
  \label{eq:rfp}
\end{equation}
For the {\em PQP} Langevin integrator we have
\begin{eqnarray*}
  \dH & = & \rational1{12}\left\{-p^3S''' + 3pS'S''\right\}\,\dt^3
    + \rational1{24}\left\{-p^4S^{(4)}\right. \\
    & & \left.+ 3p^2(2S'S'''+{S''}^2) - 3{S'}^2S''\right\}\,\dt^4 + O(\dt^5),
\end{eqnarray*}
and
\begin{equation}
  \delta\DS = p\DS'\,\dt + \left\{\half p^2\DS'' - \half S'\DS'\right\}\,\dt^2
    + O(\dt^3).
\label{eq:dDS-Lang}
\end{equation}
If we expand the integrand of equation~(\ref{eq:rfp}) we obtain to leading
non-vanishing order in~\(\dt\)
\begin{displaymath}
  \left\langle \dH + \delta\DS \right\rangle_p \asymp 0,
\end{displaymath}
and thus
\begin{eqnarray*}
  0 & \asymp & -\half\left\{S'\DS' - \DS''\right\}\dt^2 \\
  && -\rational18\left\{-{S''}^2+S''{S'}^2+S^{(4)}-2S'''S'\right\}\,\dt^4 \\ 
  && + O(\dt^6).
\end{eqnarray*}
This differential equation has an integrating factor of \(e^{-S}\),
\begin{displaymath}
  \left(\left[\DS' + \quarter(S''S'-S''')\,\dt^2\right] e^{-S}\right)' 
  \asymp 0,
\end{displaymath}
so
\begin{displaymath}
  \DS' + \quarter(S''S'-S''')\,\dt^2 \asymp Ke^S
\end{displaymath}
for some constant~\(K\).  As we are only determining the asymptotic expansion
of \(\DS\) we set \(K=0\), whence we obtain
\begin{displaymath}
  \DS' \asymp \quarter\left\{S'''-S'S''\right\}\,\dt^2.
\end{displaymath}
Up to an constant, which is fixed by the normalization of \(e^{-(S+\DS)}\),
this gives
\begin{displaymath}
  \DS \asymp \rational18\left\{2S''-{S'}^2\right\}\,\dt^2 + O(\dt^4)
\end{displaymath}
for the equilibrium distribution corresponding to the {\em PQP} Langevin
equation.

It is most important to realise that this is only an asymptotic expansion, and
the exact solution for \(\DS\) may also involve terms which are exponentially
small in the integration step size~\(\dt\).  One way to understand this is to
observe that equation~(\ref{eq:rfp}) is a Gaussian integral over the
momenta~\(p\), and that the domain of integration must include momenta \(p\gg
1/\dt\).  In particular this means that the statement made in
\cite{batrouni85a} that the shift in the equilibrium distribution corresponds
only to the addition of irrelevant operators, and that therefore the \(\dt\)
errors can be ignored if one computes quantities in the continuum limit, is
erroneous because of the existence of these subleading relevant contributions.

If we write the asymptotic expansion as \(\DS = \sum_{n\geq2} \DS_n\,\dt^n\)
then the preceding calculation is easily extended to find that the next term
satisfies the equation
\begin{equation}
  \DS_4' \asymp
    \rational1{48}\left(S'''{S'}^2 + 2S''S''' - 2S^{(4)}S' + S^{(5)}\right).
  \label{eq:DS4}
\end{equation}
While this does not give a closed-form expression for \(\DS_4\) for arbitrary
\(S\), it is obvious that if \(S\) is a polynomial in \(q\) then \(\DS_4\) and
all the other \(\DS_n\) are too.  In general the fields \(q\) have multiple
components, and equation~(\ref{eq:DS4}) then becomes
\begin{displaymath}
  (\DS_4)_i \asymp \rational1{48}\left(S_{ijk}S_jS_k + 2S_{jk}S_{ijk}
    - 2S_{ijjk}S_k + S_{ijjkk}\right)
\end{displaymath}
with the obvious notation.

For the {\em QPQ} integrator the shift in the equilibrium distribution is
\begin{eqnarray*}
  \DS_2 &=& \rational18({S'}^2 - S''), \\
  \DS_4' &=& \rational1{384}\left(24{S''}^2S' + 8S'''{S'}^2 - 20S''S'''
  \right. \\
    & & \left. - 4S^{(4)}S' - S^{(5)}\right);
\end{eqnarray*}
whereas for the {\em PQP} Campostrini wiggle (q.v., Appendix~\ref{sec:hossi})
the leading shift satisfies
\begin{eqnarray*}
  \lefteqn{\DS_4' = \rational1{288} \left(
    2(4+3\root3\of2+2\root3\of4){S''}^2S'\right.} && \\
  && - (6+4\root3\of2+5\root3\of4)S'''{S'}^2 
  - 6(4+3\root3\of2+2\root3\of4)S''S''' \\
  && \left. + 2(6+5\root3\of2+5\root3\of4)S^{(4)}S'
  - (6+5\root3\of2+5\root3\of4)S^{(5)} \right).
\end{eqnarray*}

\subsection{Hybrid Algorithm}

For the Hybrid algorithm we cannot derive an explicit formula for the shifted
equilibrium distribution in general, but we can easily see from
equation~(\ref{eq:evol}) that the leapfrog integrator conserves a Hamiltonian
\(H'\) which differs from \(H\) by terms of \(O(\dt^2)\).  This means that
\(\dH = O(\dt^2)\), and thus we deduce from equation~(\ref{eq:rfp}) that
\(\delta\DS \asymp O(\dt^2)\).  Unlike the Langevin case considered previously
the change in \(\DS\) over trajectory has no reason not to be of the same size
as \(\DS\) itself, and thus we find that \(\DS\asymp O(\dt^2)\).

\section{Noisy Algorithms}
\label{sec:noise}

One of the main advantages of the Langevin and Hybrid algorithms is that they
can be used for some non-local actions where an unbiased stochastic estimate of
the force\footnote{We do not necessarily require a noisy estimator \(\Sest\)
for the action \(S\) itself to implement the equations of motion, but we will
require one if we wish to construct an ``exact'' algorithm, q.v.
\S\ref{sec:exactnoisy}.} \(\Sest'\) can be computed relatively cheaply, with
\(\langle\Sest'(q) \rangle_\eta=S'(q)\) where the average is over some
``noise'' \(\eta\) which is chosen independently for each step.

It is helpful to think of the dynamics as that of a system evolving in the
presence of some fixed background noise field; we shall leave this dependence
on \(\eta\) implicit except where necessary.  The noisy force estimator
corresponds to the discrete mapping \(e^{tP}\) of equation~(\ref{eq:PQmap})
becoming
\begin{displaymath}
  e^{tP}: f(q,p) \mapsto f(q,p - t\Sest'(q)),
\end{displaymath}
and we thus find that the symmetric {\em QPQ} integrator conserves the
Hamiltonian \(H' \defn H + \Delta H\) where\footnote{There are two
contributions to \(\Delta H\): the first due to the use of an approximate
integrator in the presence of the background noise field, the second due to the
noise itself.}
\begin{displaymath}
  \Delta H = \Sest - S
    + \rational1{24}\left[-p^2\Sest''+2{\Sest'}^2\right]\,\dt^2
    + O(\dt^4).
\end{displaymath}
This means that the change in \(H\) over a trajectory \((q_i,p_i) \mapsto
(q_f,p_f)\) is
\begin{widetext}
\begin{eqnarray*}
  \dH &=& H(q_f,p_f) - H(q_i,p_i)
    = [H(q_f,p_f) -  H(q_i,p_i)] - [H'(q_f,p_f) - H'(q_i,p_i)] \\
    &=& [H(q_f,p_f) - H'(q_f,p_f)] - [H(q_i,p_i) - H'(q_i,p_i)]
    = - \Delta H(q_f,p_f) + \Delta H(q_i,p_i) \\
    &=& -\delta\Delta H = \delta S - \delta\Sest
      - \rational1{24}\left[-p^2\delta\Sest''
        + 2\delta\left(\Sest'\right)^2 \right]\,\dt^2 + O(\dt^4) \\
    &=& (S'-\Sest')p\,\dt
      + \half\left[{\Sest'}^2 - S'\Sest' - p^2(\Sest'' - S'')\right]\dt^2 \\
    && \qquad + \rational1{24}\left[p^3(4S''' - 3\Sest''')
        + 6p(2\Sest''\Sest - S'\Sest'' - 2S''\Sest')\right]\dt^3
      + O(\dt^4)
\end{eqnarray*}
where the last line is for a single leapfrog step where we have used
equation~(\ref{eq:QPQevol}).
\end{widetext}

If we consider the quantity \(e^{-\dH}\) averaged over the noise \(\eta\) for a
single leapfrog step we obtain
\begin{eqnarray}
  \left\langle e^{-\dH}\right\rangle_\eta
    &=& 1 - \langle\dH\rangle_\eta + \half\langle\dH^2\rangle_\eta + \ldots
      \nonumber \\
    &=& 1 - \half\left\langle(\Sest' - S')^2\right\rangle_\eta (1-p^2)\dt^2
      + O(\dt^3).
      \label{eq:noiseav}
\end{eqnarray}
Observe that the coefficient of \((1-p^2)\dt^2\) is proportional to the
variance of the noisy estimator \(\Sest'\), and thus can only vanish if the
force is computed exactly.  Thus, regardless of the order of the integrator
used there will always be an \(O(\dt^2)\) leading error for a single MD step.

\subsection{Noisy Langevin Algorithm}

If we now use equation~(\ref{eq:rfp}) to compute the equilibrium distribution
for the noisy Langevin algorithm we must average equation~(\ref{eq:noiseav})
over a Gaussian distribution for \(p\), thus we find that \[\left\langle
e^{-\dH}\right\rangle_{\eta,p} = 1 + A\,\dt^4 + O(\dt^6),\] where
\begin{eqnarray*}
  A \defn \rational1{16} \biggl[-8S'\langle\Sest'\Sest''\rangle_\eta
    + 4\langle{\Sest''}^2\rangle_\eta - 2S''\langle{\Sest'}^2\rangle_\eta &&\\ 
      + 8{S'}^2S'' + 4\langle\Sest'''\Sest'\rangle_\eta - S'S''' & & \\ 
      - S^{(4)} - 2{S'}^4 - 2{S''}^2
      + 2{S'}^2\langle{\Sest'}^2\rangle_\eta && \biggr].
\end{eqnarray*}
Since, for a noisy force equation~(\ref{eq:dDS-Lang}) is now of the form
\begin{displaymath}
  \delta\DS = p\DS'\,\dt + \half(p^2\DS'' - \Sest'\DS')\dt^2 + O(\dt^3),
\end{displaymath}
we immediately find that \(\DS = O(\dt^2)\), and more specifically it satisfies
the equation \(e^S\left(e^{-S} \DS'\right)' = -2A\,\dt^2\).

\subsection{Noisy Hybrid Algorithm}
\label{sec:noisyHybrid}

In the case of the noisy Hybrid algorithm (for which there are \(\trjlen/\dt\)
steps per momentum refreshment) we recall that the evolution is to be averaged
over independent noise for each integration step.  Since the momentum
distribution is no longer Gaussian after taking a leapfrog step, the leading
order term is not cancelled as it would be if we only did one MD step per
trajectory.  The fixed point condition now becomes
\begin{eqnarray*}
\lefteqn{ \left\langle e^{-\delta(H+\Delta S)}
\right\rangle_{p,\eta_1,...,\eta_N} } \qquad \\ & = & \left\langle
\left\langle e^{-\delta H_1} \right\rangle_{\eta_1}...\left\langle
e^{-\delta H_N}\right\rangle_{\eta_N} e^{-\delta\Delta
S}\right\rangle_p = 1,
\end{eqnarray*}
where \(\delta H=\delta H_1 +\cdots+\delta H_N\).  Thus the leading error of
the noisy Hybrid algorithm is \(\Delta S\sim O(\trjlen\dt)\), as each step
contributes an error of \(O(\dt^2)\) by equation~(\ref{eq:noiseav}).  A way of
reducing this error to \(O(\trjlen\dt^2)\) is the \(R\) algorithm
\cite{gottlieb87a}, which we shall discuss in~\S\ref{sec:R}.

Note that when we construct a Hybrid algorithm with a noisy force estimator
there is no reason to expect there to be a nearby conserved ``shadow''
Hamiltonian, as averages of non-linear Poisson brackets will not be correct.
Thus if the average error per step is \(O(\dt^2)\) we expect the errors at the
end of a trajectory of \(\trjlen/\dt\) steps to be \(O(\trjlen\dt)\) and not
just \(O(\dt)\).

\subsection[Pseudofermions and the Phi Algorithm]{Pseudofermions and the
  \(\Phi\) Algorithm}

We now turn our attention to the specific case of gauge theories in the
presence of dynamical fermions.  We start with the probability distribution for
the gauge field \(U\) with the quadratic fermion contribution integrated out in
favour of the fermionic determinant, \(P(U) \propto e^{-\Sg(U)}\,\det\M(U)\),
where \(\Sg\) is the pure gauge part of the action and \(\M\) is the fermion
kernel.  For the purpose of this discussion we ignore the pure gauge
contribution to the action since this is a simple local quantity whose force
can be computed exactly, and shall focus on the awkward fermion determinant.
As usual we may replace the determinant with an integral over a complex
pseudofermion \cite{Weingarten:1981hx} field \(\Phi\) and write the joint
probability distribution of the gauge and pseudofermion fields as \(P(U,\Phi)
\propto \exp\left\{- \left[\Sg(U) + \Phi^\dagger \M^{-1}\Phi \right]\right\}
\defn e^{-\Seff(U, \Phi)}\).  We have taken the fermion kernel to be \(\M \defn
M^\dagger M\) representing two flavours of fermion with Dirac operator \(M\) to
allow for a simple implementation of the pseudofermion heatbath.

\subsubsection[Phi Algorithm]{\(\Phi\) Algorithm} The gauge field \(U\)
corresponds to the variable \(q\) used in the general discussion before, and we
will evolve it along a classical trajectory in the presence of a fixed
pseudofermion background field \(\Phi\).  We introduce a conjugate momentum
field \(p\) in order to define a Hamiltonian \(H(U,p)\), \(P(U,p) \propto \exp
\left\{ -\left[ \half p^2 + \Seff(U)\right]\right\} \defn e^{-H(U,p)}\); the
action \(\Seff(U)\) takes the role of the potential in the Hamiltonian, and of
course depends implicitly on the background pseudofermion field.  Note that the
pseudofermion force term is given by
\begin{equation}
  \dd{\Sf}U = \Phi^\dagger \dd{\M^{-1}}U \Phi 
    = - \Phi^\dagger \M^{-1} \dd{\M}U \M^{-1} \Phi.
  \label{eq:phi}
\end{equation}

The \(\Phi\) algorithm of Gottlieb \emph{et al.} \cite{gottlieb87a} is
identical to the Hybrid algorithm described in \S\ref{sec:hybrid} with the
addition of pseudofermion refreshment from a Gaussian heatbath before each MD
trajectory.  With regards to fermions, the \(\Phi\) algorithm is restricted to
describing \(\Nf\) degenerate flavours of fermions, where \(\Nf\) is the number
of fermions described by the operator~\(\M\).

The \(\Phi\) algorithm uses a QPQ integrator, as the evaluation of the
pseudofermion force acting on the gauge fields required for the P step is only
evaluated once in this case.

According to the general arguments presented in \S\ref{sec:equilibrium} the
fixed point distribution of the \(\Phi\) algorithm must satisfy
equation~(\ref{eq:rfp}).  Performing an asymptotic expansion of this in powers
of \(\dt\), and observing that \(\delta H = O(\dt^2)\) for any trajectory
length \(\trjlen\), we deduce that \(\delta\Delta S \sim O(\dt^2)\).  We thus
have that \(\Delta S \sim O(\dt^2)\), so the \(\Phi\) algorithm is accurate to
this order.

\subsubsection[chi Algorithm]{\(\chi\) Algorithm} A slight variation of the
\(\Phi\) algorithm is the \(\chi\) algorithm.  The difference is that in the
latter the pseudofermion heatbath refreshment is performed before every MD step
as opposed to before each MD trajectory in the former.  We might expect the
error per step to be \(O(\dt^2)\), leading to an error per trajectory of
\(O(\trjlen\dt)\).  However, the error per trajectory is in fact
\(O(\trjlen\dt^2)\),\footnote{Similar to the \(\Phi\) algorithm except that it
grows with \(\trjlen\) because there is no shadow Hamiltonian.} a proof of this
will follow from that of the \(R\) algorithm to be given in~\S\ref{sec:R}.

We note in passing that the pseudofermion force for the \(\chi\) algorithm is
given by
\begin{equation}
  \dd{\Sf}U = - \left\langle\chi^\dagger \M^{-1} \dd{\M}U \M^{-1}\chi
    \right\rangle_\chi. 
  \label{eq:chi}
\end{equation}

\subsection{Non-local Actions and Noisy Hybrid Algorithms}

There is considerable interest in having the number of fermion flavours unequal
to that described by the fermion kernel (e.g., less than four flavours for
staggered fermions).  For such theories the required probability distribution
is given by \(P(U) \propto e^{-\Sg(U)} \det\M(U)^n\), where the number of
multiplets \(n\) determines the number of fermion flavours (e.g.,
\(n=\half\Nf\) for Wilson fermions and \(n=\quarter\Nf\) for staggered
fermions).  For \(n\not\in\Z\), the conventional pseudofermion approach fails
because neither a non-integer power of the Dirac operator nor its derivative
can be evaluated directly, which would be required to calculate the force.

\subsubsection[R0 Algorithm]{\(R_0\) Algorithm} 
An alternative to the pseudofermion approach is to rewrite the determinant in
trace log form where the effective fermion action is \(\Sf = -n\tr\ln\M.\) The
pseudofermionic force is replaced by a noisy estimator for the trace, since
computing the trace exactly is prohibitively expensive.  This force is written
as
\begin{eqnarray*}
  \dd{\Sf}U &=& -n \tr\left[\dd{\ln\M}U\right]
    = -n \tr\left[\M^{-1}\dd{\M}U\right]\\
    &=& -n \tr\left[(M^\dagger M)^{-1} \dd{\M}U\right] \\
    &=& -n \tr\left[{M^\dagger}^{-1} \dd{\M}U M^{-1} \right]\\
    &=& -n \left\langle \eta^\dagger{M^\dagger}^{-1}\dd{\M}U M^{-1}\eta
           \right\rangle_\eta,
\end{eqnarray*}
where \(\eta\) is a complex noise vector sampled from a Gaussian heatbath of
unit variance.  Defining an auxilliary field \(\chi \defn M^\dagger\eta\) the
force becomes
\begin{equation}
  \dd{\Sf}U = -n \left\langle\chi^\dagger \M^{-1} \dd{\M}U \M^{-1} \chi 
    \right\rangle_{\chi = M^{\dagger}\eta}.
  \label{eq:r0}
\end{equation}
With this formulation we can use the noisy Hybrid algorithm of
\S\ref{sec:noise} to generate gauge field configurations corresponding to any
number of flavours: the noisy estimator for the force \(\Sigma'\) being defined
by \(\dd{\Sf}U \defn -n \left\langle \Sigma' \right\rangle_\chi\).  This is the
\(R_0\) algorithm of Gottlieb \emph{et al.}~\cite{gottlieb87a} and has leading
order error \(O(\trjlen\dt)\).

\subsection[Reversibility, Area Preservation, and the R
  Algorithm]{Reversibility, Area Preservation, and the \(R\) Algorithm}
\label{sec:R}

As can be seen from equations~(\ref{eq:phi}), (\ref{eq:chi}), and
(\ref{eq:r0}), the \(\Phi\), \(\chi\), and \(R_0\) algorithms have similar
``pseudofermion'' force terms despite their different derivations.  Indeed, we
introduced the \(\chi\) algorithm to emphasise this similarity: in both the
\(\chi\) and \(R_0\) algorithms the ``pseudofermion'' force is computed from
Gaussian noise \(\eta\) for each MD step.  The difference is that in the former
the pseudofermion field \(\chi = M^\dagger\Bigl(U(t)\Bigr) \eta\) is calculated
from a heatbath at the beginning of each MD step, whereas in the latter the
auxilliary field \(\chi = M^\dagger\Bigl(U(t+\half\dt)\Bigr) \eta\) is
calculated at the midpoint of each MD step (that is, at the same time as the
force itself is evaluated for the integrator).

The \(\chi\) algorithm has \(O(\trjlen\dt^2)\) errors for \(n=1\) multiplets,
whereas the \(R_0\) algorithm has \(O(\trjlen\dt)\) errors.  However, for
\(n=0\) multiplets (i.e., no fermions) the algorithms are identical and have
errors of \(O(\dt^2)\).  It seems reasonable to expect that the leading error
has a linear dependence both on the time within the MD step at which the
pseudofermions are generated from their heatbath and on the number of
multiplets, so if we evaluate the pseudofermion field at time \(t =
\half(1-n)\dt\) through the MD step for \(0\le n\le 1\) we should obtain an
\(O(\trjlen\dt^2)\) algorithm.  This is the \(R\) algorithm of Gottlieb
\emph{et al.}  \cite{gottlieb87a}.  For two flavours of staggered fermions,
this means evaluating the pseudofermion field a quarter way through each MD
update.  Note that this algorithm is neither reversible nor area-preserving.

To prove that the \(R\) algorithm does indeed have \(O(\trjlen\dt^2)\) leading
order error we again look at the fixed point (equilibrium) distribution.  The
condition is that given in equation~(\ref{eq:req}), and in this case we have
neither area-preservation nor reversibility.  We consider a single step of the
\(R\) algorithm, where the auxilliary field \(\chi\) is computed at a time
\(t=\half(1-\alpha)\dt\) with \(\alpha\) a free parameter.  Expanding
equation~(\ref{eq:req}) to leading non-vanishing order we obtain
\[\left\langle(\delta+\dbar)\Delta S\right\rangle_{p,\eta} \sim
-\left\langle(\delta+\dbar)H + \tr\ln\,U_*\right\rangle_{p,\eta} + \cdots.\]

We can compute the leading contributions to this quantity as follows; the
change in energy over a trajectory is \[\langle \deltaH\rangle_{p,\eta} = 2 n
(n + 2p^2\alpha ) \tr\left[ \minv \dmdu \minv \dmdu\right] \dt^2,\] where the
\(O(\dt)\) term vanishes upon noise averaging.  Taylor expanding the Jacobian
for each update step gives \[\langle \tr\ln U_*\rangle_{p,\eta} = n\alpha
\tr\left[\minv\dmdu\minv \dmdu\right] \dt^2.\] Finally, the leading order
contribution to the quantity \(\dbar U = U^{-1}(\dt) - F\comp U\comp F\) that
measures the lack of reversibility of the integrator is \[\langle \dbar
H\rangle_{p,\eta} = 2 p^2 n \alpha \tr\left[\minv \dmdu \minv \dmdu\right]
\dt^2.\]

We thus find that \[\langle(\delta+\dbar)\Delta S\rangle_{p,\eta} \sim - A\dt^2
+ O(\dt^3)\] where \[A = \half n(n-\alpha)(1-p^2) \tr \left[\minv\dmdu\minv
\dmdu \right].\] If we choose \(\alpha=n\) the leading term cancels, and thus
the leading error is \(O(\trjlen\dt^2)\) for an entire trajectory of
\(\trjlen/\dt\) steps.  Therefore, as claimed, the \(R\) algorithm has errors
of \(O(\trjlen\dt^2)\), and thus so does the \(\chi\) algorithm since it
corresponds to the special case of the \(R\) algorithm with \(n=1\).

\subsection{Exact Noisy Algorithms}
\label{sec:exactnoisy}

It is interesting to consider whether the noisy algorithms described above can
be made exact (in the sense of having no integrator step-size errors), and if
so how.

With the addition of an accept/reject step after the MD step, the \(\Phi\)
algorithm becomes the exact Hybrid Monte Carlo (HMC) algorithm \cite{duane87a}.

The noisy Hybrid algorithm can be made exact by including a noisy acceptance
step \cite{kennedy85e,kennedy85f,Lin:1999qu,Bakeyev:2000wm} after each MD
integration step.  Note that a trajectory is defined as being the MD evolution
between momentum refreshments, and can consist of any number of MDMC steps,
where an MDMC step is an MD step followed by a (noisy) acceptance test.  Using
just one acceptance test at the end of the molecular dynamics trajectory is not
valid, since reversibility is violated because of the noise.  Unfortunately
this exact version of the algorithm suffers from a significantly longer
autocorrelation time: this is because the momentum must be flipped if a
rejection occurs at any of the \(\trjlen/\dt\) noisy acceptance tests required
per trajectory, just as for the second-order Langevin (Kramer's)
algorithm~\cite{kennedy91b,kennedy91a,Kennedy:1999di}.

The \(\chi\) algorithm can be made exact by the addition of a Metropolis
acceptance step after each MD step, and including a momentum flip after every
rejected MD update: however, the resulting algorithm suffers from the same
problems as the exact noisy Hybrid algorithm discussed before.\footnote{In
fact, it is essentially the same algorithm.}

The \(R_0\) algorithm can be made exact through the addition of a noisy
acceptance test, however, since the algorithm has scaling \(O(\trjlen\dt)\),
this would require very small step sizes.  Unfortunately the \(R\) algorithm
cannot be made exact by adding a (noisy) Metropolis acceptance step, because
the lack of area-preservation and reversibility preclude detail balance being
satisfied.

Clearly, the use of exact algorithms is preferable to that of the inexact
algorithms described in this work.  For non-local actions (e.g., a non-integral
number of fermion multiplets) the Rational Hybrid Monte Carlo algorithm
\cite{Clark:2006fx,Clark:2006wp} is a good candidate: in this algorithm the
fractional power of the fermion kernel that appears in the pseudofermion
bilinear is replaced by a rational approximation that can be directly evaluated
and differentiated.  Hence a pseudofermion formulation can be used, and a
Metropolis acceptance test can be added to render the algorithm exact.

\section{Conclusions}

In this paper we have described how the use of a symplectic integrator for the
numerical integration of Hamilton's equations of motion not only preserves the
reversibility and area-preservating properties of their exact solution, but
also conserves a ``shadow'' Hamiltonian close to the original one.  We have
shown how this conserved Hamiltonian may be computed as a power series in the
integration step size \(\dt\) using the BCH formula for the Lie algebra of
Poisson brackets.

We then considered Markov processes of the Hybrid type, in which molecular
dynamics and momentum refreshment Markov steps are alternated.  Since the fixed
point of the molecular dynamics step does not coincide with that of the
momentum heatbath neither can be the fixed point of the full Markov process.
We have derived a general condition (equation~(\ref{eq:req})) for this fixed
point distribution, which simplifies to equation~(\ref{eq:rfp}) for the case of
reversible and area preserving MD integrators, and shown how properties of this
distribution can be found by expanding these conditions in powers of the
integration step size \(\dt\).  For the case of the Langevin algorithm, it was
shown how the equilibrium distribution can be found explicitly to any order in
\(\dt\), and why this is only an asymptotic expansion.

Finally we considered those algorithms which use a noisy estimate of the force.
Here the leading order behaviour of these algorithms was found for the noisy
Langevin and Hybrid algorithms.  It was shown that in general noisy Hybrid
algorithms have leading order error \(O(\trjlen\dt)\) regardless of the order
of the numerical integrator, however, there are special cases where we can
cancel this leading term through a judicious choice of when we evaluate the
noise, i.e., the \(R\) algorithm.  We also considered how to render these
algorithms exact through the addition of a Metropolis acceptance test.

\section*{Acknowledgements}

ADK would like to thank Stefan Sint, Ivan Horv\'ath, Barry Trager, and Jos\'e
Figueroa-O'Farrill for helpful discussions.  MAC is supported under NSF
grant PHY-0427646.

\pagebreak
\appendix\section{Hamiltonian Vector Fields and Poisson Brackets}
\label{sec:poisson}

We shall denote by \(\Lambda^k\) the set of antisymmetric multilinear
\(k\)-forms that act on \(k\)-tuples of vectors in the tangent bundle \(T\M\)
over a manifold \(\M\).
\begin{definition}
  The \emph{exterior derivative} \(d:\Lambda^k\to\Lambda^{k+1}\) is the unique
  linear transformation satisfying
  \begin{enumerate}
    \item \(df(v) = v(f)\) for any \(0\)-form \(f\in\Lambda^0\) and vector
      field \(v\in T\M\),
    \item \(d^2 = 0\), and
    \item \(d(\alpha \wedge \beta) = (d\alpha)\wedge\beta +
    (-1)^{\deg\alpha}\alpha\wedge d\beta\) (the anti-Leibniz rule).
  \end{enumerate}
\end{definition}

\begin{lemma}
  For any \(\Omega\in\Lambda^k\)
  \begin{eqnarray*}
    d\Omega(v_0,\ldots,v_k) 
    = \sum_i (-)^i v_i\Omega(v_0, \ldots, {\hat v}_i, \ldots, v_k) & & \\
    + \sum_{i<j} (-)^{i+j} \Omega([v_i,v_j], v_0, \ldots, {\hat v}_i,
        \ldots, {\hat v}_j, \ldots, v_k),& & 
\end{eqnarray*}
where \({\hat v}_j\) indicates that the term \(v_j\) is omitted, and
\([a,b]\in T\M\) is the commutator of two vector fields \(a, b\in T\M\).
\end{lemma}

\begin{widetext}
\begin{proof}
  We may express \(\Omega\) in a local coordinate patch as \(\Omega =
  \frac1{k!} \Omega_{\mu_1,\ldots,\mu_k} dq^{\mu_1}\wedge\ldots \wedge
  dq^{\mu_k}\), and thus
  \begin{eqnarray*}
    \lefteqn{d\Omega(v_0,\ldots,v_k)
      = \frac1{k!} \dd{\Omega_{\mu_1,\ldots,\mu_k}}{q^{\mu_0}}
	dq^{\mu_0} \wedge dq^{\mu_1} \wedge \ldots \wedge dq^{\mu_k}
	(v_0,\ldots,v_k)} \qquad && \\
      &=& \sum_{i=0}^k (-)^i
	\dd{\Omega_{\mu_0,\ldots,{\hat \mu}_i,\ldots,\mu_k}}{q^{\mu_i}}
	  v_0^{\mu_0} v_1^{\mu_1}\cdots {\hat v}_i^{\mu_i}\cdots v_k^{\mu_k} \\
      &=& \sum_{i=0}^k (-)^i v_i^{\mu_i} \dd{}{q^{\mu_i}}
	\left[\Omega_{\mu_0,\ldots,{\hat\mu}_i,\ldots,\mu_k}
	  v_0^{\mu_0} \cdots {\hat v}_i^{\mu_i} \cdots v_k^{\mu_k}\right] \\
      && \quad + \sum_{0\leq i<j\leq k} (-)^{i+j}
	\Omega_{\mu_0,\ldots,{\hat\mu}_i,\ldots,\mu_k}
	  \left[v_i^{\mu_i} \dd{v_j^{\mu_j}}{q^{\mu_i}}
	     -  v_j^{\mu_j} \dd{v_i^{\mu_i}}{q^{\mu_j}}\right]
	  v_0^{\mu_0} \cdots {\hat v}_i^{\mu_i} \cdots {\hat v}_j^{\mu_j}
	  \cdots v_k^{\mu_k} \\
      &=& \sum_i (-)^i v_i\Omega(v_0, \ldots, {\hat v}_i, \ldots, v_k) +
	\sum_{i<j} (-)^{i+j} \Omega([v_i,v_j], v_0, \ldots, {\hat v}_i,
	  \ldots, {\hat v}_j, \ldots, v_k).
    \end{eqnarray*}
  The result is independent of the coordinate system used for this
  verification.
\end{proof}
\end{widetext}

For any \(1\)-form \(\theta\) and \(2\)-form \(\omega\) this identity is
\begin{eqnarray}
  d\theta(a,b) &=& a\theta(b) - b\theta(a) - \theta([a,b]); \nonumber \\
  d\omega(a,b,c) &=& a\omega(b,c) + b\omega(c,a) + c\omega(a,b)
  \label{eq:d2form}\\
    & & - \omega([a,b],c) - \omega([b,c],a) - \omega([c,a],b). \nonumber
\end{eqnarray}

\begin{definition}
  The cotangent bundle \(T^*\M\) has a \emph{symplectic structure} if there is
  a non-singular closed \emph{fundamental \(2\)-form} \(\omega\).
\end{definition}

\begin{definition}
  For each \(0\)-form \(F\) on \(T^*\M\) there is a corresponding
  \emph{Hamiltonian vector field} \(\hat F\) defined by \(dF \defn i_{\hat F}
  \omega\) where \(i\) is the \emph{interior produce}; equivalently we may
  write this as \(dF(x) = \omega(\hat F,x)\;\forall x\in T\M\).
\end{definition}

\begin{definition}
  The \emph{Poisson bracket} of two \(0\)-forms is \[\{A,B\} \defn -\omega(\hat
  A, \hat B) \qquad A,B\in\Lambda^0.\] 
\end{definition}

\begin{lemma} \label{lemma:hvf}
  The action of a Hamltonian vector field \(\hat A\) on a \(0\)-form \(F\) is
  given by \(\hat AF=\{A,F\}\).
\end{lemma}

\begin{proof}
  \(\hat AF = dF(\hat A) = i_{\hat F}\omega(\hat A) = \omega(\hat F,\hat A) =
  \{A,F\}\).
\end{proof}

\begin{lemma}
  The space of \(0\)-forms on \(T^*\M\) together with the Poisson bracket form
  a Lie algebra; that is the Poisson bracket satisfies
  \begin{itemize}
    \item \(\{A,A\} = 0\) and
    \item \(\{A,\{B,C\}\} + \{B,\{C,A\}\} +\{C,\{A,B\}\} = 0\) \\(Jacobi
    identity).
  \end{itemize}
\end{lemma}

\begin{proof}
  The antisymmetry is obvious.  To establish the Jacobi identity consider
  \(d\omega(\hat A,\hat B,\hat C)\) for three arbitrary Hamiltonian vector
  fields \(\hat A\), \(\hat B\) and \(\hat C\).  Recalling that the fundamental
  \(2\)-form is closed, \(d\omega = 0\), and using equation~(\ref{eq:d2form})
  we have
  \begin{eqnarray}
    d\omega(\hat A,\hat B,\hat C)
      &=& \hat A\omega(\hat B,\hat C) + \hat B\omega(\hat C,\hat A)
        + \hat C\omega(\hat A,\hat B) \nonumber \\
      && + \omega([\hat A,\hat B],\hat C)
	+ \omega([\hat B,\hat C],\hat A) \nonumber\\
      && + \omega([\hat C,\hat A],\hat B)\nonumber\\
	& = & 0. \label{eq:domega}
  \end{eqnarray}
  Now, \(\hat A\omega(\hat B,\hat C) = -\hat A\{B,C\}\) by the definition of
  the Poisson bracket, so \(\hat A\omega(\hat B,\hat C) = -\{A,\{B,C\}\}\) by
  application of lemma~\ref{lemma:hvf}. 

  Similarly, 
  \begin{eqnarray*}
  \lefteqn{\omega([\hat A,\hat B],\hat C) = -\omega(\hat C,[\hat A,\hat B])
    = -i_{\hat C}\omega([\hat A,\hat B])} && \\ 
   & = & -dC([\hat A,\hat B]) = -[\hat A,\hat  B]C \\
    & = & - (\hat A\hat B - \hat B\hat A)C = - \hat A\{B,C\} + \hat B\{A,C\} \\
    & = & - \{A,\{B,C\}\} + \{B,\{A,C\}\}.
  \end{eqnarray*}
  The Jacobi identity follows upon
  subsitituting these relations into equation~\ref{eq:domega}.
\end{proof}

\begin{theorem}
  Let \(\hat A\) and \(\hat B\) be two Hamiltonian vector fields, then their
  commutator is also a Hamiltonian vector field and furthermore \([\hat A,\hat
  B] = \widehat{\{A,B\}}\).
\end{theorem}

\begin{proof}
  Consider the action of the commutator on an arbitrary \(F\in\Lambda^0\),
  \([\hat A,\hat B]F = (\hat A\hat B - \hat B\hat A)F = \{A,\{B,F\}\} -
  \{B,\{A,F\}\}\) upon applying lemma~\ref{lemma:hvf}.  Using the Jacobi
  identity we obtain \([\hat A,\hat B]F = - \{F,\{A,B\}\} = \{\{A,B\},F\} =
  \widehat{\{A,B\}}F\).  As this must hold \(\forall F\) the result follows.
\end{proof}

This argument may be carried out more explicitly in a local coordinate patch.
Locally \(\omega\) may always be written as \(\sum_i dq_i\wedge dp_i\) (Darboux
theorem) so 

\begin{eqnarray*}\omega(\hat A,x) & = & dA(x) \\
\implies \left(\sum_i dq_i\wedge
dp_i\right)(\hat A,x) & = & \left(\dd{A}{q^i} dq_i + \dd{A}{p^i}
dp_i\right)(x)
\end{eqnarray*}
\begin{eqnarray*}
\lefteqn{ \implies \left(\sum_i dq_i\wedge dp_i\right) \left(a^j\dd{}{q^j}
  + \tilde a^j\dd{}{p^j}, x^k\dd{}{q^k} + \tilde x^k\dd{}{p^k}\right) } \qquad
  \qquad \qquad \qquad & & \\
   &=& a^i\tilde x^i-\tilde a^i x^i = \dd{A}{q^i} x^i + \dd{A}{p^i}\tilde x^i;
\end{eqnarray*}
and as this must hold for arbitrary \(x\) we may identify \[a^i =
\dd{A}{p^i} \quad \hbox{and}\quad \tilde a^i=-\dd{A}{q^i},\] and thus
\[\hat A = \sum_i \left(\dd{A}{p^i}\dd{}{q^i} - \dd{A}{q^i}\dd{}{p^i}
\right).\] 
\begin{widetext}
The commutator of \(\hat A\) and \(\hat B\) is
\begin{eqnarray*}
  [\hat A,\hat B] &=& \left[
  \sum_i\left(\dd{A}{p^i} \dd{}{q^i} - \dd{A}{q^i} \dd{}{p^i}\right),
  \sum_j\left(\dd{B}{p^j} \dd{}{q^j} - \dd{B}{q^j} \dd{}{p^j}\right)
  \right] \\
  &=& \sum_{i,j} \left[ \left(
    \dd{A}{p^i} \ddd{B}{q^i}{p^j}
    - \dd{B}{p^i} \ddd{A}{q^i}{p^j}
    - \dd{A}{q^i} \ddd{B}{p^i}{p^j}
    + \dd{B}{q^i} \ddd{A}{p^i}{p^j}
  \right) \right. \dd{}{q^j} \\
  && \quad + \left. \left(
    - \dd{A}{p^i} \ddd{B}{q^i}{q^j}
    + \dd{B}{p^i} \ddd{A}{q^i}{q^j}
    + \dd{A}{q^i} \ddd{B}{p^i}{q^j}
    - \dd{B}{q^i} \ddd{A}{p^i}{q^j}
  \right) \dd{}{p^j} \right] \\
  &=& \sum_j \left[
    \dd{}{p^j} \sum_i 
      \left( \dd{A}{p^i}\dd{B}{q^i} - \dd{A}{q^i}\dd{B}{p^i}\right)
      \dd{}{q^j} - \dd{}{q^j}
      \left( \dd{A}{p^i}\dd{B}{q^i} - \dd{A}{q^i}\dd{B}{p^i}\right)
      \dd{}{p^j}
    \right] \\
  &=& \sum_j\left(\dd{\{A,B\}}{p^j} \dd{}{q^j} - \dd{\{A,B\}}{q^j}
  \dd{}{p^j}\right) = \widehat{\{A,B\}},
\end{eqnarray*}
where the Poisson bracket is
\begin{eqnarray*}
  \{A,B\} &\defn& -\sum_i 
    \left(\dd{A}{q^i}\dd{B}{p^i} - \dd{A}{p^i}\dd{B}{q^i}\right) \\
  &=& -\left(\sum_i dq_i\wedge dp_i\right) \left(
    \dd{A}{p^j}\dd{}{q^j} - \dd{A}{q^j}\dd{}{p^j},
    \dd{B}{p^k}\dd{}{q^k} - \dd{B}{q^k}\dd{}{p^k}
  \right) = -\omega(\hat A,\hat B).
\end{eqnarray*}
\end{widetext}
\section{Basic Properties of Lie Algebras}

\begin{definition}
  Let \(K\) be a commutative ring with unit.  A {\em Lie Algebra} over \(K\) is
  a \(K\)-module\footnote{Recall that if \(K\) is a field then a \(K\)-module
  is a linear space.} \(\L\) together with a \(K\)-bilinear mapping
  \(\L\times\L\to\L: (x,y)\mapsto[x,y]\) called a {\em Lie bracket} which
  satisifies
  \begin{equation}
    \begin{array}{rll}
      [x,x] &= 0 & \\ \relax
      [x,[y,z]] + [y,[z,x]] + [z,[x,y]] &= 0 &\hbox{(Jacobi identity)}
    \end{array}
    \label{eq:liedef}
  \end{equation}
  for all \(x,y,z\in\L\).
\end{definition}
Note that this implies that the Lie bracket is antisymmetric,
\begin{displaymath}
  [x,y] + [y,x] = 0 \qquad\forall x,y\in\L
\end{displaymath}
because \([x+y,x+y] = [x,x] + [x,y] + [y,x] + [y,y] = [x,y] + [y,x]\).  The
converse is also true unless \(K\) has characteristic~2.

If \(\A\) is an associative algebra over \(K\), then it has a natural Lie
algebra structure\footnote{Note that the Jacobi identity holds automatically.}
given by \([x,y] \defn xy - yx\).
\begin{lemma} \label{lemma:envalg}
  For any given Lie algebra \(\L\) there is a unique associative algebra
  \(\A_0\) called the {\em enveloping algebra of \(\L\)} with a Lie algebra
  homomorphism \(\phi_0:\L\to\A_0\) defined by \(\phi_0:[x,y]\mapsto xy - yx\)
  which has the universal property that for all \(\A\) for which there is a Lie
  algebra homomorphism \(\phi:\L\to\A\) there is a unique algebra homomorphism
  \(f:\A_0\to\A\) such that \(\phi = f\comp\phi_0\).
\end{lemma}
In other words the enveloping algebra ``contains'' all associative algebras
which have \(\L\) as their natural Lie algebra.
\begin{proof}
  Let \(\T\) be the tensor algebra\footnote{\(\T\) is the universal algebra
  over the module \(\L\).  We shall denote multiplication in \(\T\)
  by~\(\otimes\).} over \(\L\) and \(\I\) be the ideal of \(\T\) generated by
  the elements \(x\otimes y - y\otimes x - [x,y]\;(\forall x,y\in\L)\); then
  \(\A_0 = \T/\I\).
\end{proof}

\subsection{Poincar\'e--Birkhoff--Witt Theorem}

Not only is a Lie algebra contained in its enveloping algebra, but in fact the
tensor algebra has a direct sum decomposition \(\T = \T_s\oplus\I\), where
\(\T_s\) is the subalgebra a \(\T\) consisting of all symmetric tensors, so
\(\A_0 \cong\T_S\).  This is the content of the following
\begin{theorem}[Poincar\'e--Birkhoff--Witt]
  Let \(\L\) be a Lie algebra over \(K\) which is a free \(K\)-module with a
  totally ordered basis \((x_i)\), and let \(\A_0\) be its enveloping
  algebra.  Then \(\A_0\) is a free \(K\)-module with the set of ordered
  products \(\phi_0(x_{i_1}) \ldots \phi_0(x_{i_n}) \;(i_0\leq\ldots\leq i_n)\)
  as a basis.
\end{theorem}
{\par\em Proof.}
  Recall from the proof of Lemma~\ref{lemma:envalg} that \(\A_0=\T/\I\).  If we
  define \(\T_s\) to be the submodule of \(\T\) spanned by the ordered products
  \(x_{i_1}\otimes\cdots\otimes x_{i_n}\;(i_1\leq\ldots\leq i_n)\), then what
  we must show is that \(\T\) is the direct sum of modules \(\T_s\) and \(\I\),
  that is \(\T=\T_s\oplus\I\) and \(\T_s\cap\I=\emptyset\).  We shall do this
  by showing that each element of \(\T\) has a unique decomposition into an
  element of \(\T_s\) and an element of the ideal~\(\I\).  Clearly, for
  elements of \(T_s\) or \(\I\) we must have
  \begin{eqnarray*}
    x_{i_1}\otimes\cdots\otimes x_{i_n} &=&
    \Bigl\{x_{i_1}\otimes\cdots\otimes x_{i_n}\Bigr\} \oplus 0
    (i_1\leq\ldots\leq i_n),
  \end{eqnarray*}
  \begin{eqnarray*}
    \lefteqn{\alpha\otimes(x\otimes y - y\otimes x - [x,y])\otimes\beta} & &
    \qquad \\
    &=& 0 \oplus \Bigl\{\alpha\otimes(x\otimes y - y\otimes x - [x,y])
    \otimes\beta\Bigr\} (\alpha,\beta\in\T).
  \end{eqnarray*}
  For any other basis element of \(\T\) where the factors are not in increasing
  order we have
  \begin{eqnarray*}
   \lefteqn{ \alpha\otimes y\otimes x\otimes\beta } \qquad & & \\
      & = & \Bigl\{\alpha\otimes x\otimes y\otimes\beta - \alpha\otimes[x,y]
        \otimes\beta\Bigr\} \\
      & &   \oplus\Bigl\{
          \alpha\otimes(y\otimes x - x\otimes y + [x,y])\otimes\beta\Bigr\}.
  \end{eqnarray*}
  where \(x<y\).  All that remains to show is that this decomposition is
  well-defined and does not depend upon the order in which we apply the
  preceding identity.  The only non-trivial case occurs when there are three
  adjacent out-of-order factors, for which we have both
\begin{widetext}
  \begin{eqnarray*}
    \alpha\otimes z\otimes y\otimes x\otimes\beta
      &=& \Bigl\{\alpha\otimes z\otimes x\otimes y\otimes\beta
        - \alpha\otimes z\otimes[x,y]\otimes\beta\Bigr\}
      \oplus\;\Bigl\{
        \alpha\otimes z\otimes(y\otimes x - x\otimes y + [x,y])
          \otimes\beta\Bigr\} \\
      &=& \Bigl\{\alpha\otimes x\otimes z\otimes y\otimes\beta
        - \alpha\otimes[x,z]\otimes y\otimes\beta
        - \alpha\otimes z\otimes[x,y]\otimes\beta\Bigr\} \\
      &&\oplus\;\Bigl\{
        \alpha\otimes(z\otimes x - x\otimes z + [x,z])\otimes y
          \otimes\beta+\alpha\otimes z\otimes(y\otimes x - x\otimes y + [x,y])
          \otimes\beta\Bigr\} \\
      &=& \Bigl\{\alpha\otimes x\otimes y\otimes z\otimes\beta
        - \alpha\otimes x\otimes[y,z]\otimes\beta
	- \alpha\otimes[x,z]\otimes y\otimes\beta
        - \alpha\otimes z\otimes[x,y]\otimes\beta\Bigr\} \\
      &&\oplus\;\Bigl\{
        \alpha\otimes x\otimes(z\otimes y - y\otimes z + [y,z])
          \otimes\beta
	  + \alpha\otimes(z\otimes x - x\otimes z + [x,z])\otimes y
          \otimes\beta \\
      &&\qquad + \alpha\otimes z\otimes(y\otimes x - x\otimes y + [x,y])
          \otimes\beta\Bigr\} \\
  \end{eqnarray*}
  and
  \begin{eqnarray*}
    \alpha\otimes z\otimes y\otimes x\otimes\beta
      &=& \Bigl\{\alpha\otimes y\otimes z\otimes x\otimes\beta
        - \alpha\otimes[y,z]\otimes x\otimes\beta\Bigr\} 
	\oplus\;\Bigl\{
        \alpha\otimes(z\otimes y - y\otimes z + [y,z])\otimes x
          \otimes\beta\Bigr\} \\
      &=& \Bigl\{\alpha\otimes y\otimes x\otimes z\otimes\beta
        - \alpha\otimes y\otimes[x,z]\otimes\beta
        - \alpha\otimes[y,z]\otimes x\otimes\beta\Bigr\} \\
      &&\oplus\;\Bigl\{
        \alpha\otimes y\otimes(z\otimes x - x\otimes z + [x,z])
          \otimes\beta + \alpha\otimes(z\otimes y-y\otimes z + [y,z])\otimes x
          \otimes\beta\Bigr\} \\
      &=& \Bigl\{\alpha\otimes x\otimes y\otimes z\otimes\beta
        - \alpha\otimes[x,y]\otimes z\otimes\beta 
	- \alpha\otimes y\otimes[x,z]\otimes\beta
        - \alpha\otimes[y,z]\otimes x\otimes\beta\Bigr\} \\
      &&\oplus\;\Bigl\{
        \alpha\otimes(y\otimes x - x\otimes y + [x,y])\otimes z
          \otimes\beta  + \alpha\otimes y\otimes(z\otimes x-x\otimes z + [x,z])
          \otimes\beta \\
      &&\qquad + \alpha\otimes(z\otimes y - y\otimes z + [y,z])\otimes x
          \otimes\beta\Bigr\}.
  \end{eqnarray*}
  These two values differ by
  \begin{eqnarray*}
    &&\Bigl\{\alpha\otimes x\otimes[y,z]\otimes\beta
      + \alpha\otimes[x,z]\otimes y\otimes\beta
      + \alpha\otimes z\otimes[x,y]\otimes\beta \\
    &&\qquad\qquad\quad - \alpha\otimes[x,y]\otimes z\otimes\beta
      - \alpha\otimes y\otimes[x,z]\otimes\beta
      - \alpha\otimes[y,z]\otimes x\otimes\beta\Bigr\} \\
    &&\qquad\quad\oplus\;\Bigl\{\alpha\otimes x\otimes[y,z]\otimes\beta
      + \alpha\otimes[x,z]\otimes y\otimes\beta
      + \alpha\otimes z\otimes[x,y]\otimes\beta \\
    &&\qquad\qquad\quad - \alpha\otimes[x,y]\otimes z\otimes\beta
      - \alpha\otimes y\otimes[x,z]\otimes\beta
      - \alpha\otimes[y,z]\otimes x\otimes\beta\Bigr\} \\
    &&\quad \:=\: \Bigl\{
      \alpha\otimes(x\otimes[y,z] - [y,z]\otimes x)\otimes\beta
      + \alpha\otimes([x,z]\otimes y - y\otimes[x,z])\otimes\beta 
      + \alpha\otimes(z\otimes[x,y] - [x,y]\otimes z)
      \otimes\beta\Bigr\} \\
    &&\qquad\quad\oplus\;\Bigl\{\alpha\otimes[x,z]\otimes y\otimes\beta
        + \alpha\otimes z\otimes[x,y]\otimes\beta
        - \alpha\otimes[x,y]\otimes z\otimes\beta - \alpha\otimes
      y\otimes[x,z]\otimes\beta
        - \alpha\otimes[y,z]\otimes x\otimes\beta\Bigr\} \\
    &&\quad \:=\: \Bigl\{
        \alpha\otimes([x,[y,z]] + [y,[z,x]] + [z,[x,y]])\otimes\beta\Bigr\}
      \oplus\;\Bigl\{ 
          \alpha\otimes([x,[y,z]] + [y,[z,x]] + [z,[x,y]])\otimes\beta\Bigr\}
      \:=\: 0
  \end{eqnarray*}
  where we have used the Jacobi identity.
\badbreak$\;\bsquare$\smallskip
\end{widetext}

\subsection{Free Lie Algebras}

We are interested in the properties of the Lie algebra generated by some set of
operators, but {\em a priori} we know nothing about the nature of the Lie
brackets of these generators.  We therefore wish to work in the context of the
most general Lie algebra which can be constructed from these generators, for
which all Lie brackets are assumed distinguishable unless they are related by
the defining relations~(\ref{eq:liedef}): any further relations between Lie
brackets may be applied {\em post facto}.  More formally this means that we
wish to carry out our calculations in the {\em free Lie algebra} over our set
of generators.

Let \(\L_0\) be a Lie Algebra over \(K\), and \(A\) a set with a mapping
\(i:A\to\L_0\). \(\L_0\) is called {\em free} on \(A\) if for any Lie algebra
\(\L\) and any mapping \(f:A\to\L\) there is a unique Lie algebra homomorphism
\(\bar f:\L_0\to\L\) such that \(\bar f\comp i=f\).  This is a universal
definition, as shown by
\begin{theorem}
  For every set \(A\) there is a unique free Lie algebra \(\L(A)\) on \(A\),
  \(\L(A)\) is a naturally graded \(K\)-module, \(i\) is an injection, the
  component of \(\L(A)\) of degree~1 is the free submodule generated by
  \(i(A)\), and \(\L(A)\) is generated as a Lie algebra by~\(A\).
\end{theorem}

\subsection{Hall Bases}

Central to calculations in free Lie algebras is the question of how many
independent basis elements are there of a given degree, and how to reduce an
arbitrary expression to canonical form in terms of such a basis.
\begin{definition}[Hall Trees~\cite{reutenauer93a,hall33a,magnus37a,hall50a}]
  Given a set \(A\) consider the set of all binary trees with leaf nodes
  labelled by elements of \(A\): this set is called the {\em free Magma}
  \(M(A)\).  We shall denote the tree \(h\) whose left subtree is \(h'\) and
  whose right subtree is \(h''\) by~\(h=(h',h'')\).  A {\em Hall set} \(H\) is
  a totally ordered subset of \(M(A)\) containing \(A\) which satisfies
  \begin{displaymath}
    \begin{array}{c}
      h < h''\qquad\forall h=(h',h'')\in H-A \\
      h=(h',h'')\in H\\\iff\left\{
      \begin{array}{l}
        \hbox{\(h',h''\in H\) and \(h'<h''\) and} \\
        \hbox{either \(h'\in A\) or \(h'=(x,y)\) and \(y\geq h''\).}
      \end{array}\right.
    \end{array}
  \end{displaymath}
\end{definition}
There is map \(f:M(A)\to\L(A)\) defined by \(f(a)\mapsto a\) if \(a\in A\), and
\(f:(h',h'')\mapsto[f(h'),f(h'')]\); the result of applying this map to a Hall
tree gives a {\em Hall word}, and the Hall words corresponding to any Hall set
form a {\em Hall basis} for \(\L(A)\) (as a \(K\)-module).

We shall use the following Hall basis for our calculations: \(x\in H\) iff
\begin{displaymath}
  \begin{array}{r@{\quad}l@{\qquad}l}
    & x\in A \\
    \hbox{or} & x=[y,z] & y,z\in H,\;y<z \\
    \hbox{or} & x=[y,[z,u]] & y,z,u\in H,\;z<y \\
    \hbox{and} & x<y &\hbox{if \(\deg x<\deg y\)}.
  \end{array}
\end{displaymath}
Any expression built out of Lie brackets may be reduced to canonical form by
the applying the following transformations:
\begin{displaymath}
  \begin{array}{rl@{\quad}l}
    [sx+ty,z] &\mapsto s[x,y] + t[y,z] &\hbox{where \(s,t\in K\)} \\ \relax
    [z,y] &\mapsto -[y,z] &\hbox{if \(y<z\)} \\ \relax
    [y,[z,u]] &\mapsto -[u,[y,z]] + [z,[y,u]] &\hbox{if \(y<z<u\),}
  \end{array}
\end{displaymath}
with the elements of \(A\) themselves ordered lexicographically (i.e.,
alphabetically).

For generating set \(\{A,B\}\) this leads to the following basis:
\begin{displaymath}
  \begin{array}{l}
    \Bigl\{A,\quad B\Bigr\}, \qquad
    \Bigl\{[A,B]\Bigr\}, \\[1ex]
    \Bigl\{[A,[A,B]],
      \quad [B,[A,B]]\Bigr\}, \\ [1ex]
    \Bigl\{[A,[A,[A,B]]],
      \quad [B,[A,[A,B]]],
      \quad [B,[B,[A,B]]]\Bigr\}, \\ [1ex]
    \Bigl\{[[A,B],[B,[A,B]]],
      \quad [B,[B,[B,[A,B]]]], \\ [1ex]
      \qquad\quad [A,[A,[A,[A,B]]]], 
      \quad [B,[A,[A,[A,B]]]], \\ [1ex]
      \qquad\quad [[A,B],[A,[A,B]]],
      \quad [B,[B,[A,[A,B]]]]\Bigr\}, \\
    \ldots
  \end{array}
\end{displaymath}

\subsection{Dimension of Hall Bases}

The dimension of the Hall basis of degree \(N\) on a set of cardinality \(q\)
is given by Witt's formula \cite{witt37a} \(a_N = {1\over N} \sum_{d|N} \mu(d)
q^{N/d},\) where \(\mu\) is the M\"obius function.

By the Poincar\'e--Birkhoff--Witt theorem the set of ordered (symmetric)
monomials on the independent commutators is a basis for the universal
enveloping algebra of the free Lie algebra.  A basis for words of length \(N\)
is therefore provided by symmetric products of \(n_k\) words of length \(k\)
chosen from the \(a_k\) generators of the free Lie algebra, where
\(\sum_{k\geq1}kn_k = N\).  There are exactly \((-)^{n_k}{-a_k\choose n_k} =
{a_k+n_k-1\choose n_k}\) ways of choosing these \(n_k\) words symetrically, and
this is the coefficient of \(x^k\) in the series expansion of \((1-x^k)^{-a_k}
= \sum_{n_k\geq0} {-a_k\choose n_k}(-x^k)^{n_k}\).  The total number of words
of length \(N\) is thus the coefficient of \(x^N\) in the generating function
\(g\defn\prod_{k\geq1} (1-x^k)^{-a_k}\).  On the other hand, the universal
enveloping algebra is just the free algebra on \(q\) symbols, so there are
\(q^N\) independent basis elements for words of length \(N\), and thus \(a_k\)
is determined from the equation \(g = \sum_{k=0}^\infty q^kx^k = (1-qx)^{-1}\).
Witt's solution is obtained by taking the logarithm of this equation,
\(-\sum_{k\geq1} a_k\ln(1-x^k) = -\ln(1-qx)\), and equating the coefficients of
\(x^N\), \(\sum_{d|N} a_d/(N/d) = q^N/N\).  Using the M\"obius inversion
formula we obtain \(a_N = {1\over N}\sum_{d|N} \mu(d)q^{N/d}\).

The number of independent commutators on \(q\) letters is therefore
\begin{eqnarray*}
    a_1^{(q)} &=& q, \\
    a_2^{(q)} &=& \half q(q-1), \\
    a_3^{(q)} &=& \third q(q-1)(q+1), \\
    a_4^{(q)} &=& \quarter q^2(q-1)(q+1), \\
    a_5^{(q)} &=& \rational15 q(q-1)(q+1)(q^2+1), \\
    a_6^{(q)} &=& \rational16 q(q-1)(q+1)(q^3+q-1), \\
    \ldots
\end{eqnarray*}
More specifically we have \(a_1^{(2)}=2\), \(a_2^{(2)}=1\), \(a_3^{(2)}=2\),
\(a_4^{(2)}=3\), \(a_5^{(2)}=6\), \(a_6^{(2)}=9\), \(a_7^{(2)}=18\),
\(a_8^{(2)}=30\), \(a_9^{(2)}=56\), \(a_{10}^{(2)}=99\), \(a_{11}^{(2)}=186\),
\(a_{12}^{(2)}=335\), and so on.

\subsection{Baker--Campbell--Hausdorff Formula}

\begin{theorem}[\cite{reutenauer93a}] \label{thm:dps}
  Let \(M\) be an associative algebra over a commutative ring \(K\supset\Q\)
  and let \(d\) be a derivation on \(M\).  For any power series \(f(X) =
  \sum_{n\geq0} a_nX^n\) we have
  \begin{equation}
    df(X) = \sum_{k\geq1} {1\over k!} (\ad X)^{k-1}(dX) f^{(k)}(X).
  \label{eq:dps}
  \end{equation}
  where \(\ad a: b \mapsto [a,b]\).
\end{theorem}
{\par\em Proof.}
  Equation~(\ref{eq:dps}) is linear in \(f\), so it suffices to consider \(f(X)
  = X^n\), for which we have
  \begin{displaymath}
    dX^n = \sum_{k=1}^n {n\choose k} (\ad X)^{k-1}(dX) X^{n-k}.
  \end{displaymath}
  We shall prove this by induction on \(n\): for \(n=0\) it is trivially true,
  and
\begin{widetext}
  \begin{displaymath}
    dX^{n+1} = d(X\,X^n) = dX\,X^n + X\,dX^n
      = dX\,X^n + X \sum_{k=1}^n {n\choose k} (\ad X)^{k-1}(dX) X^{n-k}.
  \end{displaymath}
  Using the identity \(Xu = [X,u] + uX = (\ad X)u + uX\) we obtain
  \begin{eqnarray*}
    &&dX^{n+1} = dX\,X^n + \sum_{k=1}^n {n\choose k} (\ad X)^k(dX) X^{n-k}
        + \sum_{k=1}^n {n\choose k} (\ad X)^{k-1}(dX) X^{n+1-k} \\
    &&\qquad = dX\,X^n
      + \sum_{k=2}^{n+1} {n\choose k-1} (\ad X)^{k-1}(dX) X^{n+1-k}
        + \sum_{k=1}^n {n\choose k} (\ad X)^{k-1}(dX) X^{n+1-k} \\
    &&\qquad = dX\,X^n + \sum_{k=2}^n \left[{n\choose k-1} + {n\choose
        k}\right] (\ad X)^{k-1}(dX) X^{n+1-k} + (\ad X)^n(dX) + n\,dX X^n \\
    &&\qquad = \sum_{k=1}^{n+1} {n+1\choose k} (\ad X)^{k-1}(dX) X^{n+1-k}.
  \end{eqnarray*}
{\badbreak$\;\bsquare$\smallskip}
\end{widetext}

\begin{corollary}
  \begin{eqnarray*}
    de^X & = & g(\ad X)(dX) e^X\\
\end{eqnarray*}
\(where\)
\begin{eqnarray*}
    g(x) & = & \sum_{k\geq1} {x^{k-1}\over k!} = {e^x - 1\over x}.
  \end{eqnarray*}
  \label{cor:a}
\end{corollary}
\begin{proof}
  By Theorem~\ref{thm:dps} we have
  \begin{displaymath}
    de^X = \sum_{k\geq1} {1\over k!} (\ad X)^{k-1}(dX) e^X
      = g(\ad X)(dX) e^X.
  \end{displaymath}
\end{proof}

\begin{definition}
  We define the {\em Hausdorff series} as
  \begin{displaymath}
    H \defn \sum_{n\geq1} c_n(A,B)
    \qquad\hbox{where}\qquad
    e^H \defn e^Ae^B,
  \end{displaymath}
  and the \(c_n\) are homogeneous of degree \(n\) in the generators \(A\) and
  \(B\).
\end{definition}

\begin{theorem} \label{thm:dH}
  Using the derivations\footnote{If \(K=\R\) these definitions are equivalent
  to \(d_Af(A,B) = \left.\dd{f(tA,B)}t\right|_{t=1}\), \(d_Bf(A,B) = \left.
  \dd{f(A,tB)}t\right|_{t=1}\) and \(df(A,B) = \left.\dd{f(tA,tB)}t
  \right|_{t=1} \).} \(d_A\) and \(d_B\) defined by \(d_AA = A\), \(d_AB = 0\),
  \(d_BA=0\), and \(d_BB=B\), and setting \(d \defn d_A + d_B\), we have
  \begin{displaymath}
    dH = \left({\ad H\over2}\coth{\ad H\over2}\right)(A+B)
      - \left({\ad H\over2}\right)(A-B).
  \end{displaymath}
\end{theorem}
For any derivation \(D\) we have \(D1 = D(1\cdot1) = D1\cdot1 + 1\cdot D1 =
2\,D1\), so \(D1 = 0\).  It is also trivial to verify inductively that \(d_XX^n
= nX^n\) and \(d_Xe^X = Xe^X\).
\begin{proof}
  Using Corollary~\ref{cor:a} we obtain
  \begin{eqnarray*}
    d_Ae^H &=& d_A(e^Ae^B) = d_Ae^A\,e^B + e^A\,d_Ae^B\\
      & = & Ae^H = g(\ad H)(d_AH) e^H, \\
    d_Be^{-H} &=& d_B(e^{-B}e^{-A}) = d_Be^{-B}\,e^{-A} + e^{-B}\,d_Be^{-A}\\
      & = & -Be^{-H} = g\Bigl(\ad(-H)\Bigr)\Bigl(d_B(-H)\Bigr) e^{-H};
  \end{eqnarray*}
  whence
  \begin{eqnarray*}
    A &=& g(\ad H)(d_AH),\\
    B &=& g(-\ad H)(d_BH)
  \end{eqnarray*}
  and
  \begin{eqnarray*}
    d_AH &=& [g(\ad H)]^{-1}A, \\
    d_BH &=& [g(-\ad H)]^{-1}B.
  \end{eqnarray*}
  We now observe that \(1/g(x)\) may be decomposed into the sum of an even and
  an odd function
  \begin{displaymath}
    {1\over g(\pm x)} = {x\over2}\left(\coth{x\over2} \mp 1\right),
  \end{displaymath}
  and the desired result follows immediately since \(d = d_A + d_B\).
\end{proof}

\begin{definition}
  The {\em Bernoulli numbers} \(B_n\) are defined by
  \begin{displaymath}
    {1\over g(x)} = {x\over e^x - 1} \defn \sum_{n\geq0} {B_nx^n\over n!}.
  \end{displaymath}
\end{definition}
Taking the symmetric part of this expression we see that
\begin{displaymath}
  {x\over2}\coth{x\over2} = \sum_{m\geq0} {B_{2m}x^{2m}\over(2m)!}.
\end{displaymath}

\begin{widetext}
\begin{theorem}[\cite{baker05a,campbell97a,campbell98a,hausdorff06a,czyz94a,%
varadarajan74a}]
  The terms in the Hausdorff series \(c_n(A,B)\) are given by the recursion
  relations
  \begin{displaymath}
    c_{n+1} = {1\over n+1} \biggl\{-\half[c_n,A-B]
      + \sum_{m=0}^{\lfloor n/2\rfloor} {B_{2m}\over(2m)!}\!
        \sum_{{k_1,\ldots,k_{2m}\geq1}\atop{k_1+\cdots+k_{2m}=n}}\!\!
          [c_{k_1},[\ldots,[c_{k_{2m}}, A+B]\ldots]] \biggr\}.
  \end{displaymath}
\end{theorem}

\begin{proof}
  Using the properties that the \(c_n\) are homogeneous in \(A\) and \(B\) and
  that \(\ad1=0\) we have
  \begin{eqnarray*}
    dH &=& \sum_{n\geq0} dc_n = \sum_{n\geq1} nc_n, \\
    \ad H &=& \sum_{k\geq1} \ad c_k;
  \end{eqnarray*}
  hence by Theorem~\ref{thm:dH}
  \begin{eqnarray*}
    dH &=& \left({\ad H\over2}\coth{\ad H\over2}\right)(A+B)
        - \left({\ad H\over2}\right)(A-B) \\
      &=& \sum_{m\geq0}{B_{2m}\over(2m)!}(\ad H)^{2m}(A+B)
        - \left({\ad H\over2}\right)(A-B) \\
      &=& \sum_{m\geq0} {B_{2m}\over(2m)!}
        \Bigl(\sum_{k\geq1}\ad c_k\Bigr)^{2m}(A+B)
          - \left({\ad H\over2}\right)(A-B) \\
  \noalign{\hbox{and so}}
      \sum_{n\geq0} (n+1)c_{n+1} &=& \sum_{n\geq1}
        \sum_{m=0}^{\lfloor n/2\rfloor} {B_{2m}\over(2m)!}\!
          \sum_{{k_1,\ldots,k_{2m}\geq1}\atop{k_1+\cdots+k_{2m}=n}}\!\!
            \ad c_{k_1}\ldots\ad c_{k_{2m}} (A+B)
              - \sum_{n\geq1} \half(\ad c_n)(A-B) \\
  \noalign{\hbox{therefore, equating terms of equal gradation,}}
      (n+1)c_{n+1} &=& \sum_{m=0}^{\lfloor n/2\rfloor} {B_{2m}\over(2m)!}\!
        \sum_{{k_1,\ldots,k_{2m}\geq1}\atop{k_1+\cdots+k_{2m}=n}}\!\!
          \ad c_{k_1}\ldots\ad c_{k_{2m}} (A+B)
            - \half(\ad c_n)(A-B).
  \end{eqnarray*}
\end{proof}

The first few terms in the Hausdorff series are
\begin{eqnarray}
  \ln(e^Ae^B) &=& \bigl\{A + B\bigr\} + \half[A,B]
      + \rational1{12} \Bigl\{[A,[A,B]] - [B,[A,B]]\Bigr\} \nonumber \\
    && - \rational1{24}[B,[A,[A,B]]] \nonumber \\
    && \begin{array}{lrlrll}
	 + \rational1{720}\Bigl\{ 
	   & - 4 & [B,[A,[A,[A,B]]]]
	   & - 6 & [[A,B],[A,[A,B]]] & \\
	 \vphantom{\Bigl\{} 
	   & + 4 & [B,[B,[A,[A,B]]]]
	   & - 2 & [[A,B],[B,[A,B]]] & \\
	 \vphantom{\Bigl\{}
	   & - & [A,[A,[A,[A,B]]]]
	   & + & [B,[B,[B,[A,B]]]] 
	   & \Bigr\} + \cdots
       \end{array}
  \label{eq:bch}
\end{eqnarray}
From this we easily obtain the formula for a symmetric product
\begin{eqnarray}
  \ln(e^{A/2}e^Be^{A/2}) &=& \bigl\{A + B\bigr\}
      - \rational1{24} \Bigl\{ 2 [B,[A,B]] + [A,[A,B]] \Bigr\} \nonumber \\
    && \begin{array}{lrlrll}
	 + \rational1{5760} \Bigl\{
	   & 32 & [B,[B,[A,[A,B]]]] & - 16 & [[A,B],[B,[A,B]]] & \\
	 \vphantom{\Bigl\{}
	   & + 28 & [B,[A,[A,[A,B]]]] & + 12 & [[A,B],[A,[A,B]]] & \\
	 \vphantom{\Bigl\{}
	   & + 8 & [B,[B,[B,[A,B]]]] & + 7 & [A,[A,[A,[A,B]]]] 
	   & \Bigr\} + \cdots
       \end{array}
  \label{eq:bchsym}
\end{eqnarray}
\end{widetext}

Since \(e^Ae^B = e^{A+B+\delta(A,B)}\) and \(e^{-B}e^{-A} = e^{-B-A+\delta(-B,
-A)}\) we see that \(\delta(A,B) = -\delta(-B,-A)\), so under interchange of
\(A\) and \(B\) all terms of even grading in change sign, whereas those of odd
grading do not.

\section{Higher-Order Symmetric Symplectic Integrators}\label{sec:hossi}

It was observed by Campostrini~\cite{campostrini89a,forest90b} that one can
construct higher-order integrators by the following method~\cite{creutz89a}:
since \(U_0(\dt) = e^{\dt\hat H} + R_0\dt^3 + O(\dt^5)\) we observe that the
``wiggle''
\begin{displaymath}
  U_0(\epsilon) U_0(-\sigma\epsilon) U_0(\epsilon)
    = e^{\epsilon(2-\sigma)\hat H} + R_0(2-\sigma^3)\epsilon^3 + O(\dt^5)
\end{displaymath}
can be adjusted to give an integration scheme correct to \(O(\dt^5)\) by
choosing \(\sigma=\root3\of2\).  The step size may be kept fixed by taking
\(\epsilon = \dt/(2 - \sigma)\).

Naturally, this wiggle may itself be iterated to give integration schemes of
arbitrarily high order, all of which are still reversible and area-preserving.
We use the recursive definition
\begin{eqnarray*}
    \lefteqn{ U_n(\epsilon_n) U_n(-\sigma_n\epsilon_n) U_n(\epsilon_n) } \qquad
    & & \\ 
    & = &  e^{\epsilon(2-\sigma_n)\hat H} +
    R_n(2-\sigma_n^{2n+1})\epsilon_n^{2n+1} 
      + O(\dt^{2n+3})
\end{eqnarray*}
and choose \(\sigma_n=\root{2n+1}\of2\) and \(\epsilon_n =\dt/(2 - \sigma_n)\).

It was noted by Yoshida~\cite{yoshida90a} that while the lowest-order
Campostrini scheme is optimal, in the sense of requiring the fewest integration
steps, the second order Campostrini wiggle is not.  The lowest order
Campostrini scheme corresponds to the operator
\begin{eqnarray*}
  \lefteqn{ e^{\epsilon Q/2} e^{\epsilon P} e^{\epsilon Q/2}
    e^{-\sigma\epsilon Q/2} e^{-\sigma\epsilon P} e^{-\sigma\epsilon Q/2}
    e^{\epsilon Q/2} e^{\epsilon P} e^{\epsilon Q/2} } \qquad & & \\
  & = & e^{\epsilon Q/2} e^{\epsilon P}
    e^{(1-\sigma)\epsilon Q/2} e^{-\sigma\epsilon P} e^{(1-\sigma)\epsilon Q/2}
    e^{\epsilon P} e^{\epsilon Q/2}
\end{eqnarray*}
which uses precisely seven steps.  The second order wiggle uses 19 steps, while
the 15 step operator
\begin{eqnarray*}
  && U_0(u\,\dt)\,U_0(v\,\dt)\,U_0(w\,\dt)\,U_0(x\,\dt)\,\\ [0.5ex]
  && U_0(w\,\dt)\,U_0(v\,\dt)\,U_0(u\,\dt) \\ [0.5ex]
  &&\qquad  = e^{u\,\dt Q/2} e^{u\,\dt P}
    e^{(u+v)\,\dt Q/2} e^{v\,\dt P}
    e^{(v+w)\,\dt Q/2} \\ 
    &&\qquad\qquad e^{w\,\dt P} e^{(w+x)\,\dt Q/2} e^{x\,\dt P}
    e^{(w+x)\,\dt Q/2} \\
    &&\qquad\qquad e^{w\,\dt P} e^{(v+w)\,\dt Q/2}
     e^{v\,\dt P} e^{(u+v)\,\dt Q/2}\\
    &&\qquad\qquad e^{u\,\dt P} e^{u\,\dt Q/2}
\end{eqnarray*}
has errors of \(O(\dt^7)\) if the ideal defined by
\begin{eqnarray*}
    x^5 + 2w^5 + 2 v^5 + 2u^5 &=& 0, \\
    x^3 + 2w^3 + 2v^3 + 2u^3 &=& 0, \\
    x + 2w + 2v + 2u &=& 1, \\
    4xwv^3 - 2xw^3v - 2x^3vu + 4xwu^3 - 4w^3vu && \\
    - 2x^3wu - 2xw^3u - 2x^3wv + 8wvu^3 - 2xv^3u \\
    + 4xvu^3 - x^4w + x^2w^3 - 4wv^3u - 2w^3v^2 \\
    - 2v^3u^2 - 2w^3u^2 + 2wu^4 + 2vu^4 + 4v^2u^3 \\
    - 4w^4u + 4w^2u^3 - 4v^4u + 2wv^4 - x^3u^2 \\
    + xu^4 - x^4u + x^2u^3 - 4w^4v + 4w^2v^3 \\
    + x^2v^3 - x^4v +  xv^4  - x^3v^2 - x^3w^2 + xw^4 &=& 0
\end{eqnarray*}
is not empty.  This is indeed the case, as a Gr\"obner basis computation shows
that \(v\), \(w\), and \(x\) may be expressed as polynomials in \(u\), which is
a root of the following irreducible polynomial
\begin{widetext}
\begin{displaymath}
  \begin{array}{l}
     5632424294400000000u^{39} - 92200336819200000000u^{38} \\
     + 710632361410560000000u^{37} - 3437629764814080000000u^{36} \\
     + 11745037928943360000000u^{35} - 30260229452421120000000u^{34} \\
     + 61344468339328512000000u^{33} - 100894346480650176000000u^{32} \\
     + 137871545973425856000000u^{31} - 159597428255349696000000u^{30} \\
     + 159057766014056179200000u^{29} - 138323253491741289600000u^{28} \\
     + 106099417410611328000000u^{27} - 72361810116050054400000u^{26} \\
     + 44125123305044761920000u^{25} - 24138408506765309280000u^{24} \\
     + 11867110408796028480000u^{23} - 5247011965321527840000u^{22} \\
     + 2086800152523757920000u^{21} - 746466059135744064000u^{20} \\
     + 240110266627607904000u^{19} - 69431877142547472000u^{18} \\
     + 18041618760883056000u^{17} - 4210061488653312000u^{16} \\
     + 881426634100156800u^{15} - 165335574305894400u^{14} \\
     + 27731884779770400u^{13} - 4148250096765600u^{12} \\
     + 551410054740000u^{11} - 64829840769360u^{10} \\
     + 6700635295200u^9 - 604016460000u^8 \\
     + 46995575760u^7 - 3112757280u^6 \\
     + 172255032u^5 - 7756920u^4 \\
     + 273360u^3 - 7080u^2 \\
     + 120u - 1.
  \end{array}
\end{displaymath}
This polynomial has three real roots, corresponding to the numerical solutions
found by Yoshida:
\begin{displaymath}
  \begin{array}{ll}
     u = 0.7845136104 7755726381 9497633866, &
     v = 0.2355732133 5935813368 4793198233, \\
     w = - 1.1776799841 7887100694 6415596562, &
     x = 1.3151863206 8391121888 4249687935;\\ [1.5ex]
     u = 1.4398481679 7678309093 0499281479, &
     v = 0.0042606818 7079201679 9146793392, \\
     w = - 2.1322852220 0145151552 3597811357, &
     x = 2.3763527443 0775282371 7294656042; \\ [1.5ex]
     u = 1.4477825623 9929793289 7896663298, &
     v = - 2.1440353163 0538931021 3622103136, \\
     w = 0.0015288622 8424927492 2787428878, &
     x = 2.3894477832 4368421218 6399178641.
   \end{array}
\end{displaymath}
\end{widetext}

\iftrue

\else
\bibliography{adk,lattice-bibliography,book}

\begin{thebibliography}{53}
\expandafter\ifx\csname natexlab\endcsname\relax\def\natexlab#1{#1}\fi
\expandafter\ifx\csname bibnamefont\endcsname\relax
  \def\bibnamefont#1{#1}\fi
\expandafter\ifx\csname bibfnamefont\endcsname\relax
  \def\bibfnamefont#1{#1}\fi
\expandafter\ifx\csname citenamefont\endcsname\relax
  \def\citenamefont#1{#1}\fi
\expandafter\ifx\csname url\endcsname\relax
  \def\url#1{\texttt{#1}}\fi
\expandafter\ifx\csname urlprefix\endcsname\relax\def\urlprefix{URL }\fi
\providecommand{\bibinfo}[2]{#2}
\providecommand{\eprint}[2][]{\url{#2}}

\bibitem[{\citenamefont{Weingarten and Petcher}(1981)}]{Weingarten:1981hx}
\bibinfo{author}{\bibfnamefont{D.~H.} \bibnamefont{Weingarten}}
  \bibnamefont{and} \bibinfo{author}{\bibfnamefont{D.~N.}
  \bibnamefont{Petcher}}, \bibinfo{journal}{Phys. Lett.}
  \textbf{\bibinfo{volume}{B99}}, \bibinfo{pages}{333} (\bibinfo{year}{1981}).

\bibitem[{\citenamefont{Fucito et~al.}(1981)\citenamefont{Fucito, Marinari,
  Parisi, and Rebbi}}]{pseudofermions}
\bibinfo{author}{\bibfnamefont{F.}~\bibnamefont{Fucito}},
  \bibinfo{author}{\bibfnamefont{E.}~\bibnamefont{Marinari}},
  \bibinfo{author}{\bibfnamefont{G.}~\bibnamefont{Parisi}}, \bibnamefont{and}
  \bibinfo{author}{\bibfnamefont{C.}~\bibnamefont{Rebbi}},
  \bibinfo{journal}{Nucl. Phys.} \textbf{\bibinfo{volume}{B180 [FS2]}},
  \bibinfo{pages}{369} (\bibinfo{year}{1981}).

\bibitem[{\citenamefont{Batrouni et~al.}(1985)\citenamefont{Batrouni, Katz,
  Kronfeld, Lepage, Svetitsky, and Wilson}}]{batrouni85a}
\bibinfo{author}{\bibfnamefont{G.~G.} \bibnamefont{Batrouni}},
  \bibinfo{author}{\bibfnamefont{G.~R.} \bibnamefont{Katz}},
  \bibinfo{author}{\bibfnamefont{A.~S.} \bibnamefont{Kronfeld}},
  \bibinfo{author}{\bibfnamefont{G.~P.} \bibnamefont{Lepage}},
  \bibinfo{author}{\bibfnamefont{B.}~\bibnamefont{Svetitsky}},
  \bibnamefont{and} \bibinfo{author}{\bibfnamefont{K.~G.}
  \bibnamefont{Wilson}}, \bibinfo{journal}{Phys. Rev.}
  \textbf{\bibinfo{volume}{D32}}, \bibinfo{pages}{2736} (\bibinfo{year}{1985}).

\bibitem[{\citenamefont{Callaway and Rahman}(1982)}]{callaway82a}
\bibinfo{author}{\bibfnamefont{D.~J.~E.} \bibnamefont{Callaway}}
  \bibnamefont{and} \bibinfo{author}{\bibfnamefont{A.}~\bibnamefont{Rahman}},
  \bibinfo{journal}{Phys. Rev. Lett.} \textbf{\bibinfo{volume}{49}},
  \bibinfo{pages}{613} (\bibinfo{year}{1982}).

\bibitem[{\citenamefont{Polonyi and Wyld}(1983)}]{polonyi83a}
\bibinfo{author}{\bibfnamefont{J.}~\bibnamefont{Polonyi}} \bibnamefont{and}
  \bibinfo{author}{\bibfnamefont{H.~W.} \bibnamefont{Wyld}},
  \bibinfo{journal}{Phys. Rev. Lett.} \textbf{\bibinfo{volume}{51}},
  \bibinfo{pages}{2257} (\bibinfo{year}{1983}), \bibinfo{note}{erratum: ibid.
  52:401, 1984}.

\bibitem[{\citenamefont{Duane}(1985)}]{duane85a}
\bibinfo{author}{\bibfnamefont{S.}~\bibnamefont{Duane}},
  \bibinfo{journal}{Nucl. Phys.} \textbf{\bibinfo{volume}{B257 [FS14]}},
  \bibinfo{pages}{652} (\bibinfo{year}{1985}).

\bibitem[{\citenamefont{Duane et~al.}(1987)\citenamefont{Duane, Kennedy,
  Pendleton, and Roweth}}]{duane87a}
\bibinfo{author}{\bibfnamefont{S.}~\bibnamefont{Duane}},
  \bibinfo{author}{\bibfnamefont{A.~D.} \bibnamefont{Kennedy}},
  \bibinfo{author}{\bibfnamefont{B.~J.} \bibnamefont{Pendleton}},
  \bibnamefont{and} \bibinfo{author}{\bibfnamefont{D.}~\bibnamefont{Roweth}},
  \bibinfo{journal}{Phys. Lett.} \textbf{\bibinfo{volume}{195B}},
  \bibinfo{pages}{216} (\bibinfo{year}{1987}).

\bibitem[{\citenamefont{Andersen}(1980)}]{andersen80a}
\bibinfo{author}{\bibfnamefont{H.~C.} \bibnamefont{Andersen}},
  \bibinfo{journal}{J. Chem. Phys} \textbf{\bibinfo{volume}{72}},
  \bibinfo{pages}{2384} (\bibinfo{year}{1980}).

\bibitem[{\citenamefont{Rossky et~al.}(1978)\citenamefont{Rossky, Doll, and
  Friedman}}]{rossky78a}
\bibinfo{author}{\bibfnamefont{P.~J.} \bibnamefont{Rossky}},
  \bibinfo{author}{\bibfnamefont{J.~D.} \bibnamefont{Doll}}, \bibnamefont{and}
  \bibinfo{author}{\bibfnamefont{H.~L.} \bibnamefont{Friedman}},
  \bibinfo{journal}{J. Chem. Phys.} \textbf{\bibinfo{volume}{69}},
  \bibinfo{pages}{4628} (\bibinfo{year}{1978}).

\bibitem[{\citenamefont{Gottlieb et~al.}(1987)\citenamefont{Gottlieb, Liu,
  Toussaint, Renken, and Sugar}}]{gottlieb87a}
\bibinfo{author}{\bibfnamefont{S.}~\bibnamefont{Gottlieb}},
  \bibinfo{author}{\bibfnamefont{W.}~\bibnamefont{Liu}},
  \bibinfo{author}{\bibfnamefont{D.}~\bibnamefont{Toussaint}},
  \bibinfo{author}{\bibfnamefont{R.~L.} \bibnamefont{Renken}},
  \bibnamefont{and} \bibinfo{author}{\bibfnamefont{R.~L.} \bibnamefont{Sugar}},
  \bibinfo{journal}{Phys. Rev.} \textbf{\bibinfo{volume}{D35}},
  \bibinfo{pages}{2531} (\bibinfo{year}{1987}).

\bibitem[{\citenamefont{Clark and Kennedy}(2007{\natexlab{a}})}]{Clark:2006fx}
\bibinfo{author}{\bibfnamefont{M.~A.} \bibnamefont{Clark}} \bibnamefont{and}
  \bibinfo{author}{\bibfnamefont{A.~D.} \bibnamefont{Kennedy}},
  \bibinfo{journal}{Phys. Rev. Lett.} \textbf{\bibinfo{volume}{98}},
  \bibinfo{pages}{051601} (\bibinfo{year}{2007}{\natexlab{a}}),
  \eprint{hep-lat/0608015}.

\bibitem[{\citenamefont{Clark and Kennedy}(2007{\natexlab{b}})}]{Clark:2006wp}
\bibinfo{author}{\bibfnamefont{M.~A.} \bibnamefont{Clark}} \bibnamefont{and}
  \bibinfo{author}{\bibfnamefont{A.~D.} \bibnamefont{Kennedy}},
  \bibinfo{journal}{Phys. Rev.} \textbf{\bibinfo{volume}{D75}},
  \bibinfo{pages}{011502} (\bibinfo{year}{2007}{\natexlab{b}}),
  \eprint{hep-lat/0610047}.

\bibitem[{\citenamefont{Campostrini and Rossi}(1990)}]{campostrini89a}
\bibinfo{author}{\bibfnamefont{M.}~\bibnamefont{Campostrini}} \bibnamefont{and}
  \bibinfo{author}{\bibfnamefont{P.}~\bibnamefont{Rossi}},
  \bibinfo{journal}{Nucl. Phys.} \textbf{\bibinfo{volume}{B329}},
  \bibinfo{pages}{753} (\bibinfo{year}{1990}).

\bibitem[{\citenamefont{Creutz and Gocksch}(1989)}]{creutz89a}
\bibinfo{author}{\bibfnamefont{M.}~\bibnamefont{Creutz}} \bibnamefont{and}
  \bibinfo{author}{\bibfnamefont{A.}~\bibnamefont{Gocksch}},
  \bibinfo{journal}{Phys. Rev. Lett.} \textbf{\bibinfo{volume}{63}},
  \bibinfo{pages}{9} (\bibinfo{year}{1989}).

\bibitem[{\citenamefont{Channel and Scovel}(1990)}]{channel90a}
\bibinfo{author}{\bibfnamefont{P.~J.} \bibnamefont{Channel}} \bibnamefont{and}
  \bibinfo{author}{\bibfnamefont{C.}~\bibnamefont{Scovel}},
  \bibinfo{journal}{Nonlinearity} \textbf{\bibinfo{volume}{3}},
  \bibinfo{pages}{231} (\bibinfo{year}{1990}).

\bibitem[{\citenamefont{Yoshida}(1990)}]{yoshida90a}
\bibinfo{author}{\bibfnamefont{H.}~\bibnamefont{Yoshida}},
  \bibinfo{journal}{Phys. Lett.} \textbf{\bibinfo{volume}{A150}},
  \bibinfo{pages}{262} (\bibinfo{year}{1990}).

\bibitem[{\citenamefont{Gladman
  et~al.}(1991{\natexlab{a}})\citenamefont{Gladman, Duncan, and
  Candy}}]{forest90a}
\bibinfo{author}{\bibfnamefont{B.}~\bibnamefont{Gladman}},
  \bibinfo{author}{\bibfnamefont{M.}~\bibnamefont{Duncan}}, \bibnamefont{and}
  \bibinfo{author}{\bibfnamefont{J.}~\bibnamefont{Candy}},
  \bibinfo{journal}{Celestial Mechanics and Dynamical Astronomy}
  \textbf{\bibinfo{volume}{52}}, \bibinfo{pages}{221}
  (\bibinfo{year}{1991}{\natexlab{a}}).

\bibitem[{\citenamefont{Suzuki}(1990)}]{suzuki90a}
\bibinfo{author}{\bibfnamefont{M.}~\bibnamefont{Suzuki}},
  \bibinfo{journal}{Phys. Lett.} \textbf{\bibinfo{volume}{A146}},
  \bibinfo{pages}{319} (\bibinfo{year}{1990}).

\bibitem[{\citenamefont{Suzuki}(1991)}]{suzuki91a}
\bibinfo{author}{\bibfnamefont{M.}~\bibnamefont{Suzuki}}, \bibinfo{journal}{J.
  Math. Phys.} \textbf{\bibinfo{volume}{32}}, \bibinfo{pages}{400}
  (\bibinfo{year}{1991}).

\bibitem[{\citenamefont{Gladman
  et~al.}(1991{\natexlab{b}})\citenamefont{Gladman, Duncan, and
  Candy}}]{gladman91a}
\bibinfo{author}{\bibfnamefont{B.}~\bibnamefont{Gladman}},
  \bibinfo{author}{\bibfnamefont{M.}~\bibnamefont{Duncan}}, \bibnamefont{and}
  \bibinfo{author}{\bibfnamefont{J.}~\bibnamefont{Candy}},
  \bibinfo{journal}{Celestial Mechanics and Dynamical Astronomy}
  \textbf{\bibinfo{volume}{52}}, \bibinfo{pages}{221}
  (\bibinfo{year}{1991}{\natexlab{b}}).

\bibitem[{\citenamefont{Sexton and Weingarten}(1992)}]{sexton92a}
\bibinfo{author}{\bibfnamefont{J.~C.} \bibnamefont{Sexton}} \bibnamefont{and}
  \bibinfo{author}{\bibfnamefont{D.~H.} \bibnamefont{Weingarten}},
  \bibinfo{journal}{Nucl. Phys.} \textbf{\bibinfo{volume}{B380}},
  \bibinfo{pages}{665} (\bibinfo{year}{1992}).

\bibitem[{\citenamefont{Yoshida}(1993)}]{yoshida93a}
\bibinfo{author}{\bibfnamefont{H.}~\bibnamefont{Yoshida}},
  \bibinfo{journal}{Celestial Mechanics and Dynamical Astronomy}
  \textbf{\bibinfo{volume}{56}}, \bibinfo{pages}{27} (\bibinfo{year}{1993}).

\bibitem[{\citenamefont{Hut et~al.}(1995)\citenamefont{Hut, Makino, and
  McMillan}}]{hut95a}
\bibinfo{author}{\bibfnamefont{P.}~\bibnamefont{Hut}},
  \bibinfo{author}{\bibfnamefont{J.}~\bibnamefont{Makino}}, \bibnamefont{and}
  \bibinfo{author}{\bibfnamefont{S.}~\bibnamefont{McMillan}},
  \bibinfo{journal}{The Astrophysical Journal} \textbf{\bibinfo{volume}{443}},
  \bibinfo{pages}{L93} (\bibinfo{year}{1995}).

\bibitem[{\citenamefont{Kennedy and Rossi}(1989)}]{kennedy88b}
\bibinfo{author}{\bibfnamefont{A.~D.} \bibnamefont{Kennedy}} \bibnamefont{and}
  \bibinfo{author}{\bibfnamefont{P.}~\bibnamefont{Rossi}},
  \bibinfo{journal}{Nucl. Phys.} \textbf{\bibinfo{volume}{B327}},
  \bibinfo{pages}{782} (\bibinfo{year}{1989}).

\bibitem[{\citenamefont{Duane and Kogut}(1985)}]{duane85b}
\bibinfo{author}{\bibfnamefont{S.}~\bibnamefont{Duane}} \bibnamefont{and}
  \bibinfo{author}{\bibfnamefont{J.~B.} \bibnamefont{Kogut}},
  \bibinfo{journal}{Phys. Rev. Lett.} \textbf{\bibinfo{volume}{55}},
  \bibinfo{pages}{2774} (\bibinfo{year}{1985}).

\bibitem[{\citenamefont{Duane and Kogut}(1986)}]{duane86b}
\bibinfo{author}{\bibfnamefont{S.}~\bibnamefont{Duane}} \bibnamefont{and}
  \bibinfo{author}{\bibfnamefont{J.~B.} \bibnamefont{Kogut}},
  \bibinfo{journal}{Nucl. Phys.} \textbf{\bibinfo{volume}{B275}},
  \bibinfo{pages}{398} (\bibinfo{year}{1986}).

\bibitem[{\citenamefont{Kennedy and Pendleton}(2001)}]{kennedy91a}
\bibinfo{author}{\bibfnamefont{A.~D.} \bibnamefont{Kennedy}} \bibnamefont{and}
  \bibinfo{author}{\bibfnamefont{B.~J.} \bibnamefont{Pendleton}},
  \bibinfo{journal}{Nucl. Phys.} \textbf{\bibinfo{volume}{B607 [FS]}},
  \bibinfo{pages}{456} (\bibinfo{year}{2001}), \eprint{hep-lat/0008020}.

\bibitem[{\citenamefont{Horowitz}(1987)}]{horowitz87}
\bibinfo{author}{\bibfnamefont{A.}~\bibnamefont{Horowitz}},
  \bibinfo{journal}{Nucl. Phys.} \textbf{\bibinfo{volume}{B280[FS18]}},
  \bibinfo{pages}{510} (\bibinfo{year}{1987}).

\bibitem[{\citenamefont{Horowitz}(1991)}]{horowitz90a}
\bibinfo{author}{\bibfnamefont{A.~M.} \bibnamefont{Horowitz}},
  \bibinfo{journal}{Phys. Lett.} \textbf{\bibinfo{volume}{B268}},
  \bibinfo{pages}{247} (\bibinfo{year}{1991}).

\bibitem[{\citenamefont{Kuti}(1987)}]{kuti88b}
\bibinfo{author}{\bibfnamefont{J.}~\bibnamefont{Kuti}}, in
  \emph{\bibinfo{booktitle}{Computational Physics}}, edited by
  \bibinfo{editor}{\bibfnamefont{R.~D.} \bibnamefont{Kenway}} \bibnamefont{and}
  \bibinfo{editor}{\bibfnamefont{G.~S.} \bibnamefont{Pawley}},
  \bibinfo{organization}{Scottish Universities Summer School in Physics}
  (\bibinfo{publisher}{Scottish Universities Summer School in Physics},
  \bibinfo{year}{1987}), pp. \bibinfo{pages}{311--378}.

\bibitem[{\citenamefont{Beccaria and Curci}(1994)}]{beccaria94a}
\bibinfo{author}{\bibfnamefont{M.}~\bibnamefont{Beccaria}} \bibnamefont{and}
  \bibinfo{author}{\bibfnamefont{G.}~\bibnamefont{Curci}},
  \bibinfo{journal}{Phys. Rev.} \textbf{\bibinfo{volume}{D49}},
  \bibinfo{pages}{2578} (\bibinfo{year}{1994}), \eprint{hep-lat/9307007}.

\bibitem[{\citenamefont{Beccaria et~al.}(1994)\citenamefont{Beccaria, Curci,
  and Galli}}]{beccaria94b}
\bibinfo{author}{\bibfnamefont{M.}~\bibnamefont{Beccaria}},
  \bibinfo{author}{\bibfnamefont{G.}~\bibnamefont{Curci}}, \bibnamefont{and}
  \bibinfo{author}{\bibfnamefont{L.}~\bibnamefont{Galli}},
  \bibinfo{journal}{Phys. Rev.} \textbf{\bibinfo{volume}{D49}},
  \bibinfo{pages}{2590} (\bibinfo{year}{1994}), \eprint{hep-lat/9307008}.

\bibitem[{\citenamefont{Langevin}(1908)}]{langevin08a}
\bibinfo{author}{\bibfnamefont{P.}~\bibnamefont{Langevin}},
  \bibinfo{journal}{Comptes Rendus} \textbf{\bibinfo{volume}{146}},
  \bibinfo{pages}{530} (\bibinfo{year}{1908}).

\bibitem[{\citenamefont{Fokker}(1914)}]{fokker14a}
\bibinfo{author}{\bibfnamefont{A.~D.} \bibnamefont{Fokker}},
  \bibinfo{journal}{Ann. Physik} \textbf{\bibinfo{volume}{43}},
  \bibinfo{pages}{810} (\bibinfo{year}{1914}).

\bibitem[{\citenamefont{Planck}(1917)}]{planck17a}
\bibinfo{author}{\bibfnamefont{M.}~\bibnamefont{Planck}},
  \bibinfo{journal}{Sitzber. Preu\ss. Akad. Wiss.} p. \bibinfo{pages}{324}
  (\bibinfo{year}{1917}).

\bibitem[{\citenamefont{Kennedy and Kuti}(1985)}]{kennedy85e}
\bibinfo{author}{\bibfnamefont{A.~D.} \bibnamefont{Kennedy}} \bibnamefont{and}
  \bibinfo{author}{\bibfnamefont{J.}~\bibnamefont{Kuti}},
  \bibinfo{journal}{Phys. Rev. Lett.} \textbf{\bibinfo{volume}{54}},
  \bibinfo{pages}{2473} (\bibinfo{year}{1985}).

\bibitem[{\citenamefont{Bhanot and Kennedy}(1985)}]{kennedy85f}
\bibinfo{author}{\bibfnamefont{G.}~\bibnamefont{Bhanot}} \bibnamefont{and}
  \bibinfo{author}{\bibfnamefont{A.~D.} \bibnamefont{Kennedy}},
  \bibinfo{journal}{Phys. Lett.} \textbf{\bibinfo{volume}{157B}},
  \bibinfo{pages}{70} (\bibinfo{year}{1985}).

\bibitem[{\citenamefont{Lin et~al.}(2000)\citenamefont{Lin, Liu, and
  Sloan}}]{Lin:1999qu}
\bibinfo{author}{\bibfnamefont{L.}~\bibnamefont{Lin}},
  \bibinfo{author}{\bibfnamefont{K.-F.} \bibnamefont{Liu}}, \bibnamefont{and}
  \bibinfo{author}{\bibfnamefont{J.~H.} \bibnamefont{Sloan}},
  \bibinfo{journal}{Phys. Rev.} \textbf{\bibinfo{volume}{D61}},
  \bibinfo{pages}{074505} (\bibinfo{year}{2000}), \eprint{hep-lat/9905033}.

\bibitem[{\citenamefont{Bakeyev and de~Forcrand}(2001)}]{Bakeyev:2000wm}
\bibinfo{author}{\bibfnamefont{T.~D.} \bibnamefont{Bakeyev}} \bibnamefont{and}
  \bibinfo{author}{\bibfnamefont{P.}~\bibnamefont{de~Forcrand}},
  \bibinfo{journal}{Phys. Rev.} \textbf{\bibinfo{volume}{D63}},
  \bibinfo{pages}{054505} (\bibinfo{year}{2001}), \eprint{hep-lat/0008006}.

\bibitem[{\citenamefont{Kennedy and Pendleton}(1991)}]{kennedy91b}
\bibinfo{author}{\bibfnamefont{A.~D.} \bibnamefont{Kennedy}} \bibnamefont{and}
  \bibinfo{author}{\bibfnamefont{B.~J.} \bibnamefont{Pendleton}}, in
  \emph{\bibinfo{booktitle}{Lattice '90}}, edited by
  \bibinfo{editor}{\bibfnamefont{U.~M.} \bibnamefont{Heller}},
  \bibinfo{editor}{\bibfnamefont{A.~D.} \bibnamefont{Kennedy}},
  \bibnamefont{and}
  \bibinfo{editor}{\bibfnamefont{S.}~\bibnamefont{Sanielevici}}
  (\bibinfo{year}{1991}), vol. \bibinfo{volume}{B20} of
  \emph{\bibinfo{series}{Nuclear Physics (Proceedings Supplements)}}, pp.
  \bibinfo{pages}{118--121}, \bibinfo{note}{talk presented at ``Lattice '90,''
  Tallahassee}.

\bibitem[{\citenamefont{Kennedy and Pendleton}(1999)}]{Kennedy:1999di}
\bibinfo{author}{\bibfnamefont{A.~D.} \bibnamefont{Kennedy}} \bibnamefont{and}
  \bibinfo{author}{\bibfnamefont{B.~J.} \bibnamefont{Pendleton}}, in
  \emph{\bibinfo{booktitle}{Lattice '99}}, edited by
  \bibinfo{editor}{\bibfnamefont{M.}~\bibnamefont{Campostrini}},
  \bibinfo{editor}{\bibfnamefont{S.}~\bibnamefont{Caracciolo}},
  \bibinfo{editor}{\bibfnamefont{L.}~\bibnamefont{Cosmai}},
  \bibinfo{editor}{\bibfnamefont{A.}~\bibnamefont{DiGiacomo}},
  \bibinfo{editor}{\bibfnamefont{F.}~\bibnamefont{Rapuano}}, \bibnamefont{and}
  \bibinfo{editor}{\bibfnamefont{P.}~\bibnamefont{Rossi}}
  (\bibinfo{year}{1999}), vol. \bibinfo{volume}{B83--84} of
  \emph{\bibinfo{series}{Nuclear Physics (Proceedings Supplements)}}, pp.
  \bibinfo{pages}{816--818}, \bibinfo{note}{proceedings of the XVIIth
  International Symposium on Lattice Field Theory, {Pisa}, {Italy}, 29 June--3
  July 1999}, \eprint{hep-lat/0001031}.

\bibitem[{\citenamefont{Reutenauer}(1993)}]{reutenauer93a}
\bibinfo{author}{\bibfnamefont{C.}~\bibnamefont{Reutenauer}},
  \emph{\bibinfo{title}{Free {Lie} Algebras}}, vol.~\bibinfo{volume}{7} of
  \emph{\bibinfo{series}{{London} {Mathematical} {Society} {Mongraphs}, new
  series}} (\bibinfo{publisher}{Oxford University Press},
  \bibinfo{year}{1993}), ISBN \bibinfo{isbn}{0-19-853679-8}.

\bibitem[{\citenamefont{Hall}(1933)}]{hall33a}
\bibinfo{author}{\bibfnamefont{P.}~\bibnamefont{Hall}},
  \bibinfo{journal}{Proceedings of the London Mathematical Society}
  \textbf{\bibinfo{volume}{2}}, \bibinfo{pages}{29} (\bibinfo{year}{1933}).

\bibitem[{\citenamefont{Magnus}(1937)}]{magnus37a}
\bibinfo{author}{\bibfnamefont{W.}~\bibnamefont{Magnus}},
  \bibinfo{journal}{Journal {f\"ur} die {Reine} und {Angewandte} {Mathematik}}
  \textbf{\bibinfo{volume}{177}}, \bibinfo{pages}{105} (\bibinfo{year}{1937}).

\bibitem[{\citenamefont{Hall}(1950)}]{hall50a}
\bibinfo{author}{\bibfnamefont{M.}~\bibnamefont{Hall}, \bibfnamefont{Jr.}},
  \bibinfo{journal}{Proceedings of the American Mathematical Society}
  \textbf{\bibinfo{volume}{1}}, \bibinfo{pages}{575} (\bibinfo{year}{1950}).

\bibitem[{\citenamefont{Witt}(1937)}]{witt37a}
\bibinfo{author}{\bibfnamefont{E.}~\bibnamefont{Witt}},
  \bibinfo{journal}{Mathematische Annalen} \textbf{\bibinfo{volume}{177}},
  \bibinfo{pages}{152} (\bibinfo{year}{1937}).

\bibitem[{\citenamefont{Baker}(1905)}]{baker05a}
\bibinfo{author}{\bibfnamefont{H.~F.} \bibnamefont{Baker}},
  \bibinfo{journal}{Proceedings of the London Mathematical Society}
  \textbf{\bibinfo{volume}{2}}, \bibinfo{pages}{24} (\bibinfo{year}{1905}).

\bibitem[{\citenamefont{Campbell}(1897{\natexlab{a}})}]{campbell97a}
\bibinfo{author}{\bibfnamefont{J.~E.} \bibnamefont{Campbell}},
  \bibinfo{journal}{Proceedings of the London Mathematical Society}
  \textbf{\bibinfo{volume}{1}}, \bibinfo{pages}{381}
  (\bibinfo{year}{1897}{\natexlab{a}}).

\bibitem[{\citenamefont{Campbell}(1897{\natexlab{b}})}]{campbell98a}
\bibinfo{author}{\bibfnamefont{J.~E.} \bibnamefont{Campbell}},
  \bibinfo{journal}{Proceedings of the London Mathematical Society}
  \textbf{\bibinfo{volume}{1}}, \bibinfo{pages}{14}
  (\bibinfo{year}{1897}{\natexlab{b}}).

\bibitem[{\citenamefont{Hausdorff}(1906)}]{hausdorff06a}
\bibinfo{author}{\bibfnamefont{F.}~\bibnamefont{Hausdorff}},
  \bibinfo{journal}{Leipziger Berichte} \textbf{\bibinfo{volume}{58}},
  \bibinfo{pages}{19} (\bibinfo{year}{1906}).

\bibitem[{\citenamefont{Czy\.z}(1994)}]{czyz94a}
\bibinfo{author}{\bibfnamefont{J.}~\bibnamefont{Czy\.z}},
  \emph{\bibinfo{title}{Paradoxes of Measures and Dimensions Originating in
  {Felix} {Hausdorff's} Ideas}} (\bibinfo{publisher}{World Scientific},
  \bibinfo{address}{Singapore}, \bibinfo{year}{1994}), ISBN
  \bibinfo{isbn}{9810201893}.

\bibitem[{\citenamefont{Varadarajan}(1974)}]{varadarajan74a}
\bibinfo{author}{\bibfnamefont{V.~S.} \bibnamefont{Varadarajan}},
  \emph{\bibinfo{title}{{Lie} Groups, {Lie} Algebras, and their
  {Representations}}}, Series in Modern Analysis (\bibinfo{publisher}{Prentice
  Hall}, \bibinfo{address}{Englewood Cliffs, NJ}, \bibinfo{year}{1974}), ISBN
  \bibinfo{isbn}{0-13-535732-2}.

\bibitem[{\citenamefont{Forest and Ruth}(1990)}]{forest90b}
\bibinfo{author}{\bibfnamefont{E.}~\bibnamefont{Forest}} \bibnamefont{and}
  \bibinfo{author}{\bibfnamefont{R.~D.} \bibnamefont{Ruth}},
  \bibinfo{journal}{Physica} \textbf{\bibinfo{volume}{D43}},
  \bibinfo{pages}{105} (\bibinfo{year}{1990}).

\end{thebibliography}
\bibliographystyle{apsrev}
\fi

\end{document}